\documentclass[journal,twocolumn]{IEEEtran}

\usepackage{amsmath,amssymb,amsfonts,mathtools,bm}
\usepackage{graphicx}
\usepackage{subcaption}
\usepackage[caption=false,font=normalsize,labelfont=sf,textfont=sf]{subfig}
\usepackage{stfloats}
\usepackage{array,booktabs,multirow}
\usepackage{algorithm}
\usepackage{algorithmic}
\usepackage{textcomp}
\usepackage{url}
\usepackage{tikz}
\usetikzlibrary{arrows.meta,positioning}
\usepackage{cite}
\usepackage{verbatim}
\usepackage{xcolor}

\AtBeginDocument{%
  \setlength{\abovedisplayskip}{2pt plus 1pt minus 1pt}%
  \setlength{\belowdisplayskip}{2pt plus 1pt minus 1pt}%
  \setlength{\abovedisplayshortskip}{1pt plus 1pt minus 1pt}%
  \setlength{\belowdisplayshortskip}{1pt plus 1pt minus 1pt}%
  \setlength{\jot}{2pt}%
  \setlength{\topsep}{1pt plus 1pt minus 1pt}%
  \setlength{\partopsep}{0pt}%
  \setlength{\itemsep}{1pt}%
  \setlength{\parsep}{0pt}%
}
\usepackage{amsthm}

\newtheorem{theorem}{Theorem}
\newtheorem{lemma}{Lemma}

\newtheorem{corollary}{Corollary}

\theoremstyle{remark}
\newtheorem{remark}{Remark}

\begin{document}

\title{ Real-Time Reconstruction of Markov Sources over MPR Channels}

\author{Pansee~S.~Elessawy and Nikolaos~Pappas
\thanks{This work was supported by ELLIIT. A shorter version of this paper has been accepted for presentation at the 2026 IEEE Information Theory Workshop (ITW 2026).}
\thanks{The authors are with the Department of Computer and Information
Science, Link\"oping University, 581~83 Link\"oping, Sweden
(e-mail: pansee.elessawy@liu.se; nikolaos.pappas@liu.se).}}

\markboth{IEEE Transactions on Communications}%
{Elessawy \MakeLowercase{\textit{and}} Pappas: Real-Time Reconstruction and Actuation Error over MPR Channels}

\maketitle

\begin{abstract}
This paper studies the real-time reconstruction and remote
actuation of two binary Markov sources over a shared wireless
channel with multi-packet reception (MPR). Unlike many existing collision-based formulations that discard simultaneous transmissions, we exploit MPR and evaluate communication through reconstruction and actuation errors rather than raw delivery rates. We consider two sensors observing the sources and aim to find sampling policies that minimize the weighted real-time reconstruction error (RTE), or equivalently the weighted cost of actuation error (CAE) for the considered binary sources, under per-sensor sampling constraints.
We first obtain closed-form expressions for the RTE and CAE in terms of the effective update probabilities. When the sensors randomize independently, the MPR-induced update-rate map becomes bilinear, and the constrained optimization is nonconvex. We exploit the geometry of the achievable update-rate region to show that the search over Pareto-efficient independently randomized policies reduces to a finite set of one-dimensional boundary-branch searches with closed-form candidates. As a benchmark, we allow time sharing among joint sensor actions, and we prove that at most two Pareto-extreme modes are sufficient, and use this benchmark to quantify the loss caused by independent randomization.
Numerical results reveal that MPR capability alone is not sufficient; simultaneous decoding improves the task-oriented objective only when concurrent reception is reliable for both sources; otherwise, policies that avoid simultaneous transmissions can be equally effective.

\end{abstract}


\vspace{-10pt}
\section{Introduction}

\IEEEPARstart{R}{eal-time} monitoring, tracking, and remote
actuation are central to emerging cyber-physical and autonomous
systems, including industrial automation, swarm robotics,
autonomous transportation, and networked
control~\cite{salimnejad2024realtime, pappas2021goal,
gunduz2023beyond}. In such systems, sensing devices observe a
time-varying source and transmit status updates to a remote
receiver that reconstructs the state and drives an actuation
process. As these systems scale, the generated data grows rapidly
while bandwidth, channel-access opportunities, and energy remain
limited~\cite{luo2025semantic}, making it
inefficient, and often infeasible, to transmit every observation.
The key question is therefore not how to deliver \emph{more} data,
but how to deliver the \emph{right} data at the right time, so that
the reconstruction and the actuation remain effective.

Traditional communication design is largely content-agnostic,
treating all packets as equally important and optimizing rate- or
delay-oriented metrics. This is not suitable for task-oriented
systems, where the value of an update depends on what it conveys
and on how it affects the receiver's
decisions~\cite{ gunduz2023beyond,
9994683}. The \emph{semantics-aware} and \emph{goal-oriented}
paradigm addresses this gap by jointly designing the generation,
transmission, and utilization of information according to its
significance for the underlying
task~\cite{10579545,popovski2020semantic,CALVANESESTRINATI2021107930,9919752,10479470}.
Within this paradigm, several semantic metrics quantify timeliness
and relevance. The Age of Information
(AoI)~\cite{KostaPappasAngelakis2017,YatesSunBrownKaulModianoUlukus2021,8000687}
measures freshness but is agnostic to content; the Age of Incorrect
Information (AoII)~\cite{maatouk2020aoii,maatouk2023aoii,11071330,9162726,10620879}
couples staleness with estimation error; and the Version Age of
Information (VAoI)~\cite{yates2021gossip} counts how many versions
the receiver lags behind the source. State- and context-aware
extensions~\cite{stamatakis2019control, zhou2020urgency}. A common limitation,
however, is that these metrics treat all reconstruction errors as
equally costly, which is rarely the case in control and actuation.

To capture the state-dependent consequences of erroneous
actuation, the CAE was introduced in~\cite{pappas2021goal} and
subsequently adopted~\cite{salimnejad2024realtime,
salimnejad2023stateaware} as a task-oriented distortion measure
that assigns possibly asymmetric, non-commutative costs to the
mismatch states between the source and its reconstruction.
Together with the related real-time reconstruction error, it
directly reflects the impact of stale or incorrect information on
the receiver's decisions, and is thus a natural objective for
remote tracking for actuation.

Semantics-aware joint sampling and transmission policies for
real-time tracking of Markov 
sources~\cite{salimnejad2024realtime, salimnejad2023stateaware}
show that change-aware policies can outperform
content-agnostic ones, and that transmitting fewer but more
significant updates can be optimal. Fundamental limits of remote
estimation under communication constraints were studied
in~\cite{ChakravortyMahajan2017,8812616}.

The multi-source setting,
in which an agent must schedule updates from several Markov 
sources over a constrained channel, has been
addressed in~\cite{luo2025semantic,11450399,10547339,
9518209} through constrained Markov  decision process (CMDP) formulations, structural analyses of the optimal policy, and low-complexity online algorithms. Pull-based remote tracking with correlated or partial observations has been studied in
~\cite{ZakeriMoltafetCodreanu2025,11195539,9676636}.

AoI has been analyzed in multiple-access systems\cite{
8006544} with multi-packet reception (MPR) \cite{9365698}, and unreliable wireless networks with scheduling for weighted AoI minimization \cite{KadotaSinhaUysalSinghModiano2018,8734015}. Recent works have studied status updating over ALOHA-based random access channels \cite{9358219,mono}. The works ~\cite{10949089,dlr224073} studied two-state Markov  source monitoring including joint model estimation with unknown source statistics and analytical renewal--reward characterizations under random and reactive access policies.
Multi-source AoI over shared queues was characterized in~\cite{8469047}. 

In this paper, we study two independent binary Markov  sources that share a wireless multiple-access channel with MPR. Unlike collision-based access, MPR allows the receiver to decode more than one simultaneous transmission with nonzero probability. The update probability of each source depends not only on its own sampling decisions, but also on the sampling decisions of the other sensor. This coupling creates an interaction between physical-layer reception and task-oriented reconstruction and actuation performance.
\textit{While prior work has mainly studied goal-oriented sampling over point-to-point links or AoI in random access, the role of MPR in real-time reconstruction of multiple sources and actuation has not been thoroughly investigated}. The main contributions are as follows.
\begin{itemize}
  \item We derive closed-form expressions for the steady-state RTE and
  CAE of each binary Markov  source under synchronize-or-hold estimation.

  \item We formulate the sampling-constrained independent MPR
  optimization problem. We then derive a budget-aware
  total-success envelope that separates dominant-transmitter and
  cooperative MPR operating regimes.

  \item We establish a Pareto-reduction result for independent policies:
  a Pareto-efficient policy does not require both sensors to randomize
  between the two sources. This reduces the policy search to a finite set
  of one-dimensional boundary-branch problems with closed-form candidates.

  \item We introduce a coordinated MPR time-sharing benchmark and prove that an optimal coordinated policy requires time sharing between at most two Pareto-extreme joint-action modes. This benchmark quantifies the performance loss caused by independent randomization.

  \item Numerical results compare collision, capture, and MPR channels, showing that MPR is useful only when simultaneous reception is reliable for both sources.
\end{itemize}

\vspace{-11px}

\section{System Model and Task-Oriented Metrics}
\label{sec:system_model}

We consider a time-slotted remote monitoring system with two sensors, two independent binary Markov  sources, and a common receiver communicating over a shared wireless channel with multi-packet reception (MPR) as shown in Fig.~\ref{fig:one}. Time is indexed by $t\in\{0,1,2,\ldots\}$. The state of source $i\in\{1,2\}$ at slot $t$ is denoted by $X_i(t)\in\{0,1\}$ and evolves according to the transition matrix
\begin{equation}
\mathbf{P}_i =
\begin{bmatrix}
1-\alpha_i & \alpha_i \\
\beta_i & 1-\beta_i
\end{bmatrix},
\qquad \alpha_i,\beta_i\in(0,1),
\label{eq:sys_Pi}
\end{equation}
where $\alpha_i$ and $\beta_i$ are the transition probabilities from state $0$ to state $1$ and from state $1$ to state $0$, respectively. We define the source correlation parameter
\begin{equation}
\lambda_i \triangleq 1-\alpha_i-\beta_i,
\qquad i\in\{1,2\}.
\label{eq:lambda_def}
\end{equation}
Thus, $\lambda_i>0$ corresponds to positively correlated source evolution, $\lambda_i=0$ to memoryless evolution, and $\lambda_i<0$ to negatively correlated evolution.

At each slot, sensor $k\in\{1,2\}$ selects an action $a_k(t)\in\{0,1,2\}$, where $a_k(t)=0$ denotes silence, while $a_k(t)=i$ means that sensor $k$ samples and transmits the current state of source $X_i(t)$. We focus on stationary randomized sampling policies, where
\begin{equation}
\Pr\{a_k(t)=j\}=a_{k,j},
\qquad j\in\{0,1,2\},
\label{eq:sys_random_policy_prob}
\end{equation}
with $a_{k,0}+a_{k,1}+a_{k,2}=1$ and $a_{k,j}\geq 0$.
The policy of sensor $k$ is denoted by $\mathbf a_k \triangleq (a_{k,0},a_{k,1},a_{k,2})$,
and the joint action by $\mathbf a(t)\triangleq(a_1(t),a_2(t))$. The randomized decisions are independent across sensors and slots and are not conditioned on the current source values. Let $Z_k(t)\in\{0,1\}$ denote the decoding outcome of sensor $k$, where $Z_k(t)=1$ if the packet from sensor $k$ is successfully decoded and $Z_k(t)=0$ otherwise. The MPR channel is described by four success probabilities. When only sensor $1$ transmits, its packet is decoded with probability $p_{1/1}$; when only sensor $2$ transmits, its packet is decoded with probability $p_{2/2}$. When both sensors transmit simultaneously, the packet of sensor $1$ is decoded with probability $p_{1/1,2}$ and the packet of sensor $2$ is decoded with probability $p_{2/2,1}$. Naturally, we consider
\begin{equation}
p_{1/1}\geq p_{1/1,2},
\qquad
p_{2/2}\geq p_{2/2,1}.
\label{eq:mpr_ordering}
\end{equation}
The MPR probabilities provide a physical-layer abstraction of the wireless reception process and may depend, in a specific channel model, on fading statistics, path loss, interference, and decoding thresholds~\cite{naware2005stability}.

\begin{figure*}[t]
    \centering
    \resizebox{\textwidth}{!}{%
    \begin{tikzpicture}[
        >=Stealth,
        font=\small,
        sourcebox/.style={
            draw, thick, fill=gray!2,
            minimum width=4.5cm,
            minimum height=5.8cm,
            rounded corners=2pt
        },
        innerbox/.style={
            draw, fill=white,
            minimum width=4.15cm,
            minimum height=2.6cm,
            align=center
        },
        sensor/.style={
            circle, draw, thick, fill=green!8,
            minimum size=1.5cm, align=center
        },
        mpr/.style={
            circle, draw, thick, fill=blue!5,
            minimum size=1.5cm, align=center
        },
        receiver/.style={
            circle, draw, thick,
            minimum size=1.5cm, align=center
        },
        estbox/.style={
            rectangle, draw, thick, fill=cyan!5,
            minimum width=2.15cm,
            minimum height=1.15cm,
            align=center
        },
        state/.style={circle, draw, minimum size=0.45cm, font=\scriptsize}
    ]
        \coordinate (S1_y) at (0, 1.45);
        \coordinate (S2_y) at (0, -1.45);
        \def\SensorX{5.8}
        \def\ChannelX{8.1}
        \def\ReceiverX{10.4}
        \def\EstimateX{12.9}
        \node[sourcebox] (outer) at (0,0) {};
        \node[anchor=north, yshift=-8pt, font=\bfseries] at (outer.north) {INFORMATION SOURCES};
\node[innerbox, fill=gray!15] (box1) at (S1_y) {
    \textbf{SOURCE $X_1(t)$} \\
      [0.2cm]
    \begin{tikzpicture}[baseline, nodes={state}, xscale=1.2]
        \node (s0) at (0,0) {0};
        \node (s1) at (1,0) {1};
        \path[-{Stealth}, font=\scriptsize]
            (s0) edge[loop left]  
                node[draw=none, fill=none, shape=rectangle] {$(1-\alpha_1)$} (s0)
            (s0) edge[bend left]  
                node[above, draw=none, fill=none, shape=rectangle] {$\alpha_1$} (s1)
            (s1) edge[bend left]  
                node[below, draw=none, fill=none, shape=rectangle] {$\beta_1$} (s0)
            (s1) edge[loop right] 
                node[draw=none, fill=none, shape=rectangle] {$(1-\beta_1)$} (s1);
    \end{tikzpicture} \\[0.1cm]
    \footnotesize Transition Matrix $\mathbf{P}_1$
};
\node[innerbox, fill=gray!15] (box2) at (S2_y) {
    \textbf{SOURCE $X_2(t)$} \\
[0.2cm]
    \begin{tikzpicture}[baseline, nodes={state}, xscale=1.2]
        \node (s0) at (0,0) {0};
        \node (s1) at (1,0) {1};
        \path[-{Stealth}, font=\scriptsize]
            (s0) edge[loop left]  
                node[draw=none, fill=none, shape=rectangle] {$(1-\alpha_2)$} (s0)
            (s0) edge[bend left]  
                node[above, draw=none, fill=none, shape=rectangle] {$\alpha_2$} (s1)
            (s1) edge[bend left]  
                node[below, draw=none, fill=none, shape=rectangle] {$\beta_2$} (s0)
            (s1) edge[loop right] 
                node[draw=none, fill=none, shape=rectangle] {$(1-\beta_2)$} (s1);
    \end{tikzpicture} \\[0.1cm]
    \footnotesize Transition Matrix $\mathbf{P}_2$
};
        \node[sensor] (sen1) at (\SensorX, 1.45) {\textbf{SENSOR 1} \\ \footnotesize policy $\mathbf{a}_1$};
        \node[sensor] (sen2) at (\SensorX, -1.45) {\textbf{SENSOR 2} \\ \footnotesize policy $\mathbf{a}_2$};
        \node[mpr] (mpr) at (\ChannelX, 0) {\textbf{MPR} \\ \footnotesize CHANNEL};
        \node[receiver] (rec) at (\ReceiverX, 0) {\textbf{RECEIVER} \\ $R$};
        \node[estbox] (hat1) at (\EstimateX, 1.2) {\textbf{$\hat{X}_1(t)$} \\ \scriptsize Estimate of $X_1$};
        \node[estbox] (hat2) at (\EstimateX, -1.2) {\textbf{$\hat{X}_2(t)$} \\ \scriptsize Estimate of $X_2$};
        \coordinate (exit1) at (box1.east);
        \coordinate (exit2) at (box2.east);
        \draw[-{Stealth[scale=1.1]}, thick] (exit1) -- (sen1.west)
            node[midway, above] {$a_{1,1}$};
        \draw[-{Stealth[scale=1.1]}, thick] (exit1) -- (sen2.west)
            node[pos=0.7, above right] {$a_{2,1}$};
        \draw[-{Stealth[scale=1.1]}, thick] (exit2) -- (sen1.west)
            node[pos=0.7, below right] {$a_{1,2}$};
        \draw[-{Stealth[scale=1.1]}, thick] (exit2) -- (sen2.west)
            node[midway, below] {$a_{2,2}$};
        \draw[->, ultra thick] (sen1.east) -- (mpr.145);
        \draw[->, ultra thick] (sen2.east) -- (mpr.215);
        \draw[->, ultra thick, dashed] (mpr) -- (rec);
        \draw[->, ultra thick] (rec.35) -- (hat1.west)
            node[midway, left, align=center, xshift=4pt, yshift=12pt]
            {\scriptsize $U_1(t)=1$ \\ \scriptsize (Success)};
        \draw[->, ultra thick] (rec.325) -- (hat2.west)
            node[midway, left, align=center, xshift=-4pt, yshift=-12pt]
            {\scriptsize $U_2(t)=1$ \\ \scriptsize (Success)};
\node[
    draw,
    fill=white,
    inner sep=2.5pt,
    xshift=6mm
] at (\ChannelX, 2.55) {
    \begin{tabular}{c}
        \scriptsize MPR Probabilities: \\
        \scriptsize $p_{1/1},\,p_{2/2},\,p_{1/1,2},\,p_{2/2,1}$
    \end{tabular}
};
    \end{tikzpicture}%
    }
    \caption{System model for status updates from independent binary Markov  sources through an MPR-capable channel.}
    \label{fig:one}
\end{figure*}
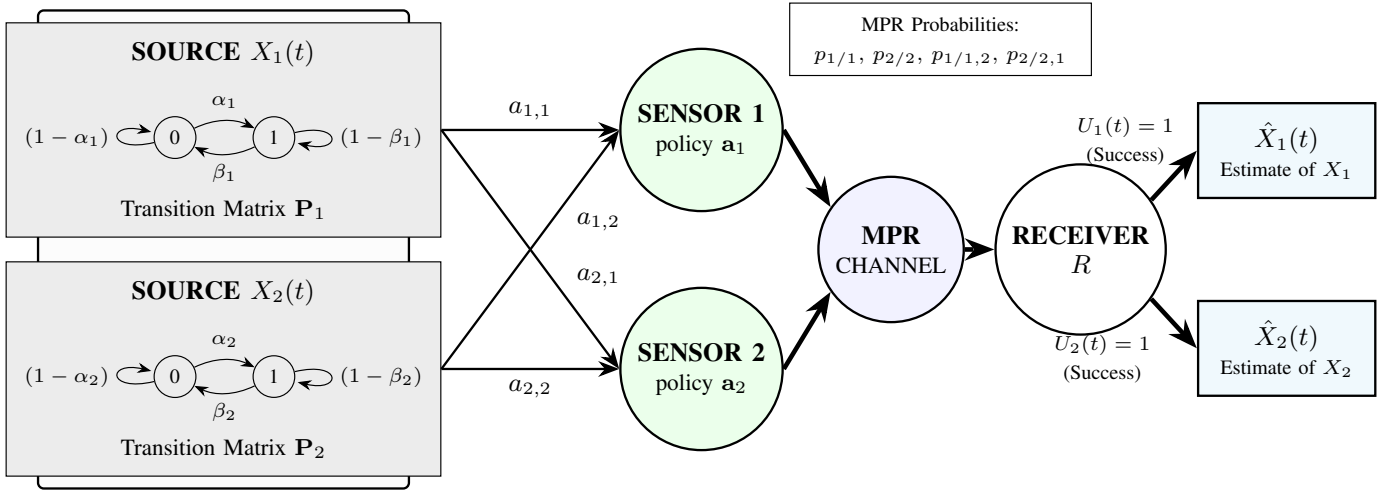
For compactness, we introduce the channel notation
\begin{equation}
s_1\triangleq p_{1/1},\qquad
s_2\triangleq p_{2/2},\qquad
a\triangleq p_{1/1,2},\qquad
b\triangleq p_{2/2,1},
\label{eq:channel_notation}
\end{equation}

let \(P_{0/\{1,2\}}\) denote the probability that no packet is decoded when sensors \(1\) and \(2\) transmit simultaneously. When both sensors transmit packets carrying information from the same source, a successful update occurs if at least one packet is decoded. Hence, we define
\begin{equation}
c \triangleq 1-P_{0/\{1,2\}} .
\label{eq:c_def}
\end{equation}

The receiver maintains an estimate $\hat X_i(t)\in\{0,1\}$ of each source. A successful update for source $i$ occurs if at least one successfully decoded packet carries the current value $X_i(t)$. Hence,
\begin{equation}
\begin{aligned}
U_i(t)\triangleq
\mathbf{1}\Big\{&
(a_1(t)=i, Z_1(t)=1)
\ \text{or}\
(a_2(t)=i, Z_2(t)=1)
\Big\},
\\
& i\in\{1,2\}.
\end{aligned}
\label{eq:sys_Ui}
\end{equation}
The receiver follows a synchronize-or-hold estimator:
\begin{equation}
\hat X_i(t)=
\begin{cases}
X_i(t), & U_i(t)=1,\\
\hat X_i(t-1), & U_i(t)=0,
\end{cases}
\qquad i\in\{1,2\}.
\label{eq:sys_sync_or_hold}
\end{equation}
The effective update probability of source $i$ is $q_i\triangleq \Pr\{U_i(t)=1\}$, $i\in\{1,2\}$.

Under the independent stationary randomized policy, the effective update probabilities are obtained by summing over the joint actions that deliver the corresponding source, giving
\begin{align}
q_1
&=
s_1 a_{1,1}a_{2,0}
+
s_2 a_{1,0}a_{2,1}
+
c a_{1,1}a_{2,1}
\nonumber\\
&\quad
+
a a_{1,1}a_{2,2}
+
b a_{1,2}a_{2,1}.
\label{eq:q1_explicit_journal}
\end{align}
Similarly, the effective update probability of source $X_2$ is
\begin{align}
q_2
&=
s_1 a_{1,2}a_{2,0}
+
s_2 a_{1,0}a_{2,2}
+
c a_{1,2}a_{2,2}
\nonumber\\
&\quad
+
a a_{1,2}a_{2,1}
+
b a_{1,1}a_{2,2}.
\label{eq:q2_explicit_journal}
\end{align}

\subsection{Real-Time Reconstruction Error}
\label{subsec:rte_journal}

We characterize the reconstruction quality of each source by studying the joint process $Y_i(t)\triangleq (X_i(t),\hat X_i(t))$, whose state space is $\mathcal S=\{(0,0),(0,1),(1,0),(1,1)\}$.
We order the states as $s_1=(0,0)$, $s_2=(0,1)$, $s_3=(1,0)$, $s_4=(1,1)$.
Within each slot, the source first evolves, and then a successfully decoded update synchronizes the receiver estimate with the current source value. If no update is decoded, the receiver keeps its previous estimate. Under this convention, $Y_i(t)$ is a four-state Markov  chain with transition matrix $\mathbf T_i$, whose entries are
\begin{equation}
T_i(k,\ell)
=
\Pr\{Y_i(t+1)=s_\ell\mid Y_i(t)=s_k\}.
\label{eq:T_entry_journal}
\end{equation}
Using \eqref{eq:sys_Pi} and \eqref{eq:sys_sync_or_hold}, direct enumeration gives
\begin{equation}
\mathbf{T}_i =
\setlength{\arraycolsep}{2.0pt}
\renewcommand{\arraystretch}{1.15}
\small
\begin{bmatrix}
1-\alpha_i & 0 & \alpha_i(1-q_i) & \alpha_i q_i \\
q_i(1-\alpha_i) & (1-\alpha_i)(1-q_i) & 0 & \alpha_i \\
\beta_i & 0 & (1-\beta_i)(1-q_i) & q_i(1-\beta_i) \\
\beta_i q_i & \beta_i(1-q_i) & 0 & 1-\beta_i
\end{bmatrix}.
\label{eq:T_matrix_journal}
\end{equation}
For example, starting from $(0,0)$, the process moves to $(1,0)$ with probability $\alpha_i(1-q_i)$ when the source changes from $0$ to $1$ but no update is received, and to $(1,1)$ with probability $\alpha_i q_i$ when the source change is followed by a successful update.

Let $\boldsymbol{\pi}_i=\big[\pi_i(0,0),\,\pi_i(0,1),\,\pi_i(1,0),\,\pi_i(1,1)\big]$
denote the stationary distribution of $Y_i(t)$. Then
\begin{equation}
\boldsymbol{\pi}_i=\boldsymbol{\pi}_i\mathbf T_i,
\qquad
\sum_{(x,\hat x)\in\mathcal S}\pi_i(x,\hat x)=1.
\label{eq:stationary_equations_journal}
\end{equation}
Following the real-time reconstruction metric in~\cite{salimnejad2024realtime}, the instantaneous reconstruction error of source $i$ is
\begin{equation}
E_i(t)\triangleq \mathbf 1\{X_i(t)\neq \hat X_i(t)\}.
\label{eq:instantaneous_rte_journal}
\end{equation}
The steady-state real-time reconstruction error (RTE) is therefore
\begin{equation}
E_i
\triangleq
\Pr\{X_i(t)\neq \hat X_i(t)\}
=
\pi_i(0,1)+\pi_i(1,0).
\label{eq:rte_def_journal}
\end{equation}

\begin{theorem}[Closed-form RTE]
\label{thm:closed_form_rte}
For source $i\in\{1,2\}$ and $q_i>0$, the two mismatch states have equal stationary probabilities,
\begin{equation}
\pi_i(0,1)=\pi_i(1,0)\triangleq \zeta_i,
\label{eq:mismatch_equality_journal}
\end{equation}
where
\begin{equation}
\zeta_i
=
\frac{\alpha_i\beta_i(1-q_i)}
{(\alpha_i+\beta_i)\big[(\alpha_i+\beta_i)-q_i(\alpha_i+\beta_i-1)\big]}.
\label{eq:zeta_closed_form_journal}
\end{equation}
Consequently, the steady-state RTE using $\lambda_i=1-\alpha_i-\beta_i$,
\begin{equation}
E_i(q_i)
=
\frac{2\alpha_i\beta_i(1-q_i)}
{(1-\lambda_i)\big[1-\lambda_i(1-q_i)\big]}.
\label{eq:Ei_closed_form_lambda_journal}
\end{equation}
\end{theorem}

\begin{IEEEproof}
Let
\begin{equation*}
\begin{aligned}
\pi_{00}&\triangleq \pi_i(0,0), &\quad
\pi_{01}&\triangleq \pi_i(0,1),\\
\pi_{10}&\triangleq \pi_i(1,0), &\quad
\pi_{11}&\triangleq \pi_i(1,1).
\end{aligned}
\end{equation*}
Since the receiver estimate does not affect the source evolution, the marginal stationary distribution of $X_i(t)$ is the stationary distribution of the binary Markov  chain in \eqref{eq:sys_Pi}. Hence,
\begin{equation}
\pi_{00}+\pi_{01}=\frac{\beta_i}{\alpha_i+\beta_i},
\qquad
\pi_{10}+\pi_{11}=\frac{\alpha_i}{\alpha_i+\beta_i}.
\label{eq:source_marginals_journal}
\end{equation}
From the stationary equations associated with the mismatch states, we have
\begin{equation}
\pi_{01}
=
(1-q_i)\big[(1-\alpha_i)\pi_{01}+\beta_i\pi_{11}\big],
\label{eq:balance_01_journal}
\end{equation}
and
\begin{equation}
\pi_{10}
=
(1-q_i)\big[\alpha_i\pi_{00}+(1-\beta_i)\pi_{10}\big].
\label{eq:balance_10_journal}
\end{equation}
Using \eqref{eq:source_marginals_journal} to eliminate $\pi_{00}$ and $\pi_{11}$ in \eqref{eq:balance_01_journal}--\eqref{eq:balance_10_journal} and subtracting the resulting equations, gives
\begin{equation}
q_i(\pi_{01}-\pi_{10})=0.
\end{equation}
Since $q_i>0$, it follows that $\pi_{01}=\pi_{10}\triangleq \zeta_i$. Substituting this equality back into either balance equation gives \eqref{eq:zeta_closed_form_journal}. Finally, since $E_i=\pi_{01}+\pi_{10}=2\zeta_i$, we obtain by rewriting the denominator using $\lambda_i=1-\alpha_i-\beta_i$ which gives \eqref{eq:Ei_closed_form_lambda_journal}.
\end{IEEEproof}

It is useful for the optimization analysis to introduce
\begin{equation}
K_i\triangleq \frac{2\alpha_i\beta_i}{1-\lambda_i},
\qquad
D_i(q_i)\triangleq 1-\lambda_i(1-q_i).
\label{eq:Ki_Di_def_journal}
\end{equation}
Then the RTE can be written compactly as
\begin{equation}
E_i(q_i)=K_i\frac{1-q_i}{D_i(q_i)}.
\label{eq:Ei_compact_journal}
\end{equation}

\begin{remark}[Zero-update-rate endpoint]
\label{rem:zero_update_endpoint}
The stationary characterization in
Theorem~\ref{thm:closed_form_rte} applies to $q_i>0$.
When $q_i=0$, no update is ever received, the estimate remains
frozen at its initial value, and the joint chain $Y_i(t)$ is
reducible. Consequently, the time-average error depends on the
initial estimate: it equals
$\alpha_i/(\alpha_i+\beta_i)$ if $\hat X_i(0)=0$, and
$\beta_i/(\alpha_i+\beta_i)$ if $\hat X_i(0)=1$.

For the optimization analysis, we define the value at zero by
continuous extension
\begin{equation}
E_i(0)
\triangleq
\lim_{q_i\downarrow0}E_i(q_i)
=
\frac{2\alpha_i\beta_i}{(\alpha_i+\beta_i)^2}.
\end{equation}
Accordingly, every occurrence of $q_i=0$ in the optimization
and candidate sets is understood in this vanishing-positive-update-rate
sense, rather than as the error of an exact zero-update process under
an arbitrary initial estimate.
\end{remark}

\begin{theorem}[Monotonicity and curvature of the RTE]
\label{thm:rte_monotonicity_curvature}
For $q_i\in[0,1]$, with $E_i(0)$ defined by the
continuous extension in Remark~\ref{rem:zero_update_endpoint}, the RTE
satisfies\begin{equation}
E_i'(q_i)
=
-\frac{K_i}{D_i(q_i)^2}<0,
\label{eq:Ei_first_derivative_journal}
\end{equation}
and
\begin{equation}
E_i''(q_i)
=
\frac{2K_i\lambda_i}{D_i(q_i)^3}.
\label{eq:Ei_second_derivative_journal}
\end{equation}
Hence, $E_i(q_i)$ is strictly decreasing in the effective update probability. Moreover, it is convex in $q_i$ when $\lambda_i>0$, linear when $\lambda_i=0$, and concave when $\lambda_i<0$.
\end{theorem}

\begin{IEEEproof}
Differentiating \eqref{eq:Ei_compact_journal} with respect to $q_i$ gives \eqref{eq:Ei_first_derivative_journal}. Differentiating once more gives \eqref{eq:Ei_second_derivative_journal}. Since $D_i(q_i)>0$ for $q_i\in[0,1]$ and $|\lambda_i|<1$, the sign of $E_i''(q_i)$ is determined by the sign of $\lambda_i$.
\end{IEEEproof}

For the two-source system, the weighted RTE is
\begin{equation}
F(q_1,q_2)
=
w_1E_1(q_1)+w_2E_2(q_2),
\qquad
w_1,w_2>0.\label{eq:weighted_rte_journal}
\end{equation}
Here, $w_i$ captures the relative semantic or actuation importance of source $i$. Since each $E_i(q_i)$ is strictly decreasing in $q_i$, increasing either effective update probability improves the corresponding reconstruction quality.

\subsection{Cost of Actuation Error}
\label{subsec:cae_journal}

We next assign costs to erroneous reconstruction states. For source $i$, let $C_i^{x,\hat x}$ denote the cost incurred when the true state is $x$ while the reconstructed state is $\hat x\neq x$. The cost of actuation error (CAE) of source $i$ is
\begin{equation}
\bar C_i
\triangleq
\sum_{x\neq \hat x}
C_i^{x,\hat x}\pi_i(x,\hat x).
\label{eq:cae_generic_journal}
\end{equation}
For the binary source, the only erroneous states are $(0,1)$ and $(1,0)$, so
\begin{equation}
\bar C_i
=
C_i^{0,1}\pi_i(0,1)
+
C_i^{1,0}\pi_i(1,0).
\label{eq:cae_binary_journal}
\end{equation}
Using Theorem~\ref{thm:closed_form_rte}, for $q_i>0$ we have
$\pi_i(0,1)=\pi_i(1,0)=\zeta_i$. Therefore,
\begin{equation}
\bar C_i
=
\big(C_i^{0,1}+C_i^{1,0}\big)\zeta_i.
\label{eq:cae_zeta_journal}
\end{equation}
Substituting \eqref{eq:zeta_closed_form_journal} gives
\begin{equation}
\bar C_i(q_i)
=
\frac{
\big(C_i^{0,1}+C_i^{1,0}\big)\alpha_i\beta_i(1-q_i)
}
{
(\alpha_i+\beta_i)\big[(\alpha_i+\beta_i)-q_i(\alpha_i+\beta_i-1)\big]
}.
\label{eq:cae_closed_form_journal}
\end{equation}
Equivalently,
\begin{equation}
\bar C_i(q_i)
=
\frac{C_i^{0,1}+C_i^{1,0}}{2}E_i(q_i).
\label{eq:cae_scaled_rte_journal}
\end{equation}
The weighted total CAE is
\begin{equation}
\bar C(q_1,q_2)
=
w_1\bar C_1(q_1)+w_2\bar C_2(q_2).
\label{eq:weighted_cae_journal}
\end{equation}

\begin{remark}
The reduction in \eqref{eq:cae_scaled_rte_journal} does not require symmetric actuation costs. The costs $C_i^{0,1}$ and $C_i^{1,0}$ may be different. The reason CAE becomes proportional to RTE is that, under the considered binary Markov  model and synchronize-or-hold estimator, the two mismatch states have equal steady-state probabilities. Therefore, weighted CAE minimization is equivalent to weighted RTE minimization with modified weights
\begin{equation}
\tilde w_i
=
w_i\frac{C_i^{0,1}+C_i^{1,0}}{2},
\qquad i\in\{1,2\}.
\label{eq:modified_weights_cae_journal}
\end{equation}
This equivalence is specific to the considered binary-source setting. For multi-state sources, or for state-dependent sampling policies whose update decisions depend on $X_i(t)$ or on the mismatch state $(X_i(t),\hat X_i(t))$, different error states need not have equal stationary probabilities, and the CAE is not generally a scaled RTE.
\end{remark}

In the remainder of the paper, we use the weighted RTE objective in \eqref{eq:weighted_rte_journal}

with the understanding that the same analysis applies to the weighted CAE objective after replacing $w_i$ by the modified weights in \eqref{eq:modified_weights_cae_journal}. 

\vspace{-10pt}
\section{Sampling-Constrained Independent MPR Optimization}
\label{sec:independent_optimization}

We now optimize the stationary independent randomized sampling policy under per-sensor sampling constraints. For each sensor $k\in\{1,2\}$, let $\Gamma_k\in(0,1]$ denote the maximum allowable sampling rate, i.e., the maximum probability with which sensor $k$ is allowed to transmit in a slot. Since sensor $k$ transmits whenever it selects action $1$ or action $2$, the sampling constraint is
\begin{equation}
a_{k,1}+a_{k,2}\leq \Gamma_k,
\qquad k\in\{1,2\}.
\label{eq:budget_constraint_ind}
\end{equation}

Using the task-oriented objective in \eqref{eq:weighted_rte_journal}, the sampling-constrained independent MPR optimization problem is

\begin{align}
F_{\mathrm{ind}}^\star(\Gamma)
&=
\min_{\{\mathbf a_k\}_{k=1}^2}
\; F(q_1,q_2)
\label{eq:Pind_objective}
\\[-0.2em]
\text{s.t.}\quad
a_{k,j}
&\geq 0,
\qquad k\in\{1,2\},\; j\in\{0,1,2\},
\label{eq:Pind_nonnegative}
\\
a_{k,0}+a_{k,1}+a_{k,2}
&= 1,
\qquad k\in\{1,2\},
\label{eq:Pind_simplex}
\\
a_{k,1}+a_{k,2}
&\leq \Gamma_k,
\qquad k\in\{1,2\}.
\label{eq:Pind_budget}
\end{align}

where $q_1$ and $q_2$ are the MPR-induced effective update probabilities in \eqref{eq:q1_explicit_journal}--\eqref{eq:q2_explicit_journal}.

To make the structure of \eqref{eq:Pind_objective}--\eqref{eq:Pind_budget} explicit, define
\begin{equation}
u_1\triangleq a_{1,1},\qquad
u_2\triangleq a_{1,2},\qquad
v_1\triangleq a_{2,1},\qquad
v_2\triangleq a_{2,2}.
\label{eq:uv_def_ind}
\end{equation}
The corresponding idle probabilities are $u_0=1-u_1-u_2$ and $v_0=1-v_1-v_2$.
The feasible set can therefore be written as
\begin{equation}
\mathcal A_\Gamma
=
\Delta_{\Gamma_1}\times \Delta_{\Gamma_2},
\label{eq:A_gamma_def}
\end{equation}
where
\begin{equation}
\Delta_{\Gamma_k}
\triangleq
\left\{
(x_1,x_2)\in\mathbb R_+^2:
x_1+x_2\leq \Gamma_k
\right\},
\qquad k\in\{1,2\}.
\label{eq:Delta_gamma_def}
\end{equation}
Thus, $\Delta_{\Gamma_k}$ is the budget-constrained action triangle of sensor $k$ in the non-idle action coordinates.

Using \eqref{eq:channel_notation}--\eqref{eq:c_def}, define \begin{equation} \gamma_{11}\triangleq c-s_1-s_2,\qquad \gamma_{12}\triangleq a-s_1,\qquad \gamma_{21}\triangleq b-s_2 . \label{eq:gamma_defs_ind} \end{equation} Since \(a\leq s_1\) and \(b\leq s_2\), \(\gamma_{12}\leq0\) and \(\gamma_{21}\leq0\). Substituting \(u_0=1-u_1-u_2\) and \(v_0=1-v_1-v_2\) into \eqref{eq:q1_explicit_journal}--\eqref{eq:q2_explicit_journal} gives \begin{equation} q_1(u,v) = s_1u_1+s_2v_1 +\gamma_{11}u_1v_1 +\gamma_{12}u_1v_2 +\gamma_{21}u_2v_1, \label{eq:q1_bilinear_ind} \end{equation} \begin{equation} q_2(u,v) = s_1u_2+s_2v_2 +\gamma_{11}u_2v_2 +\gamma_{12}u_2v_1 +\gamma_{21}u_1v_2. \label{eq:q2_bilinear_ind} \end{equation} Thus, independent randomization induces bilinear update probabilities, making the policy-space optimization generally nonconvex.

\section{Achievable Update-Rate Geometry}
\label{sec:update_rate_geometry}
The problem in Section~\ref{sec:independent_optimization} is generally
nonconvex because \(q_1(u,v)\) and \(q_2(u,v)\) are bilinear. Since the
objective depends only on the induced pair \((q_1,q_2)\), we analyze the
feasible update-rate region.

\subsection{Update-Rate Region}
\label{subsec:update_region_pareto}
Define the policy-to-update mapping
\begin{equation}
T_\Gamma:\mathcal A_\Gamma\to\mathbb R^2,
\qquad
T_\Gamma(u,v)=\big(q_1(u,v),q_2(u,v)\big),
\label{eq:TGamma_journal}
\end{equation}
where $\mathcal A_\Gamma$ is the independent feasible set in
\eqref{eq:A_gamma_def}, and $q_1,q_2$ are given by
\eqref{eq:q1_bilinear_ind}--\eqref{eq:q2_bilinear_ind}. The achievable
update-rate region is
\begin{equation}
\mathcal Q_\Gamma
\triangleq
T_\Gamma(\mathcal A_\Gamma)
=
\big\{
(q_1(u,v),q_2(u,v)):(u,v)\in\mathcal A_\Gamma
\big\}.
\label{eq:QGamma_journal}
\end{equation}
Thus, $\mathcal Q_\Gamma$ contains all update-rate pairs that can be
generated by feasible independent randomized policies.

The independent optimization problem can equivalently be written in the
update-rate plane as
\begin{equation}
F_{\mathrm{ind}}^\star(\Gamma)
=
\min_{(q_1,q_2)\in\mathcal Q_\Gamma}
F(q_1,q_2).
\label{eq:update_rate_problem_journal}
\end{equation}

\begin{lemma}[Equivalent update-rate optimization]
\label{lem:update_rate_equivalence_journal}
The policy-space problem in \eqref{eq:Pind_objective}--\eqref{eq:Pind_budget}
and the update-rate problem in \eqref{eq:update_rate_problem_journal}
have the same optimal value.
\end{lemma}

\begin{IEEEproof}
Every feasible policy $(u,v)\in\mathcal A_\Gamma$ induces exactly one
point $(q_1(u,v),q_2(u,v))\in\mathcal Q_\Gamma$. Conversely, every point
in $\mathcal Q_\Gamma$ is generated by at least one feasible policy by
definition. Since the objective depends on $(u,v)$ only through
$(q_1,q_2)$, minimizing over $\mathcal A_\Gamma$ is equivalent to
minimizing over $\mathcal Q_\Gamma$.
\end{IEEEproof}

A point \(q\in\mathcal Q_\Gamma\) is Pareto dominated if there exists
\(\widetilde q\in\mathcal Q_\Gamma\) such that
\(\widetilde q_i\ge q_i\) for \(i=1,2\), with at least one strict
inequality. The Pareto frontier \(\partial_{\mathrm P}\mathcal Q_\Gamma\)
is the set of nondominated points.

\begin{theorem}
\label{thm:pareto_localization_journal}
Every global minimizer of \eqref{eq:update_rate_problem_journal} lies on
the Pareto frontier $\partial_{\mathrm P}\mathcal Q_\Gamma$.
\end{theorem}

\begin{IEEEproof}
Suppose that $q^\star=(q_1^\star,q_2^\star)$ is a global minimizer but
is dominated by some $\widetilde q=(\widetilde q_1,\widetilde q_2)$.
Since each $E_i(q_i)$ is strictly decreasing in $q_i$, we have
\(E_i(\widetilde q_i)\le E_i(q_i^\star),\qquad i\in\{1,2\}\).
with strict inequality for at least one source. Since $w_1,w_2>0$,
\(F(\widetilde q_1,\widetilde q_2)<F(q_1^\star,q_2^\star)\).
which contradicts the optimality of $q^\star$. Hence no dominated point
can be globally optimal.
\end{IEEEproof}

\subsection{Total-Success Envelope}
\label{subsec:total_success_envelope}

We next derive a budget-aware upper bound on the total update success
$q_1+q_2$ that identifies regimes in which the Pareto frontier has a
simple structure.

The total update success can be written as
\begin{align}
q_1+q_2
&=
s_1(u_1+u_2)v_0+s_2u_0(v_1+v_2)
\nonumber\\
&\quad
+c(u_1v_1+u_2v_2)
+(a+b)(u_1v_2+u_2v_1).
\label{eq:qsum_pure_journal}
\end{align}
Let
\begin{equation}
t_1\triangleq u_1+u_2,\qquad
t_2\triangleq v_1+v_2.
\label{eq:transmission_intensities_geometry}
\end{equation}
Under the sampling constraints, $0\le t_1\le \Gamma_1$ and
$0\le t_2\le \Gamma_2$. Since
\[
u_1v_1+u_2v_2+u_1v_2+u_2v_1=t_1t_2,
\]
and $c\leq a+b$, \footnote{Conditioned on both sensors transmitting,
let $\mathcal E_k\triangleq\{Z_k(t)=1\}$. Then
$a=\Pr\{\mathcal E_1\}$, $b=\Pr\{\mathcal E_2\}$, and
$c=\Pr\{\mathcal E_1\cup\mathcal E_2\}$. Inclusion--exclusion gives
$c=a+b-\Pr\{\mathcal E_1\cap\mathcal E_2\}\leq a+b$,
$\mathcal E_k\subseteq\mathcal E_1\cup\mathcal E_2$ implies
$c\geq\max\{a,b\}$.}
we obtain
\begin{align}
q_1+q_2
&\le
s_1t_1(1-t_2)+s_2(1-t_1)t_2+(a+b)t_1t_2.
\label{eq:qsum_upper_G_journal}
\end{align}

\begin{theorem}[Budget-aware total-success envelope]
\label{thm:budget_success_envelope_journal}
For every feasible independent policy $(u,v)\in\mathcal A_\Gamma$,
\begin{equation}
q_1(u,v)+q_2(u,v)\le B_\Gamma,
\label{eq:BGamma_bound_journal}
\end{equation}
where
\begin{align}
B_\Gamma
=
\max\Big\{
&\Gamma_1s_1,\;
\Gamma_2s_2,
\nonumber\\
&
s_1\Gamma_1(1-\Gamma_2)
+s_2(1-\Gamma_1)\Gamma_2
+(a+b)\Gamma_1\Gamma_2
\Big\}.
\label{eq:BGamma_def_journal}
\end{align}
Moreover, each of the three values inside the maximum is achievable by
a feasible independent policy.
\end{theorem}

\begin{IEEEproof}
Define
\begin{equation}
G(t_1,t_2)
\triangleq
s_1t_1(1-t_2)+s_2(1-t_1)t_2+(a+b)t_1t_2.
\label{eq:G_def_journal}
\end{equation}
By \eqref{eq:qsum_upper_G_journal}, $q_1+q_2\le G(t_1,t_2)$. The function
$G(t_1,t_2)$ is bilinear in $(t_1,t_2)$. Over the rectangle
$[0,\Gamma_1]\times[0,\Gamma_2]$, it can be written as the convex
combination
\begin{equation}
\hspace*{-0.35cm}
\begin{aligned}
G(t_1,t_2)
={}&
\Big(1-\frac{t_1}{\Gamma_1}\Big)
\Big(1-\frac{t_2}{\Gamma_2}\Big)G(0,0)
+
\frac{t_1}{\Gamma_1}
\Big(1-\frac{t_2}{\Gamma_2}\Big)G(\Gamma_1,0)
\\
&+
\Big(1-\frac{t_1}{\Gamma_1}\Big)
\frac{t_2}{\Gamma_2}G(0,\Gamma_2)
+
\frac{t_1t_2}{\Gamma_1\Gamma_2}G(\Gamma_1,\Gamma_2).
\end{aligned}
\label{eq:G_weighted_average_corners}
\end{equation}
Hence $G(t_1,t_2)$ cannot exceed the maximum of its four corner values.
Evaluating the corners gives $G(0,0)=0$, $G(\Gamma_1,0)=\Gamma_1s_1$,
$G(0,\Gamma_2)=\Gamma_2s_2$, and
\begin{equation}
G(\Gamma_1,\Gamma_2)
=
s_1\Gamma_1(1-\Gamma_2)
+s_2(1-\Gamma_1)\Gamma_2
+(a+b)\Gamma_1\Gamma_2.
\end{equation}
This proves \eqref{eq:BGamma_bound_journal}.

The value $\Gamma_1s_1$ is achieved by setting $v_1=v_2=0$ and
$u_1+u_2=\Gamma_1$. The value $\Gamma_2s_2$ is achieved by setting
$u_1=u_2=0$ and $v_1+v_2=\Gamma_2$. The cooperative value
$G(\Gamma_1,\Gamma_2)$ is achieved, for example, by
\[
u_1=\Gamma_1,\quad u_2=0,\quad v_1=0,\quad v_2=\Gamma_2,
\]
for which simultaneous transmissions always carry different information from sources.
\end{IEEEproof}

\subsection{Dominant-Transmitter Regimes}
\label{subsec:dominant_transmitter_geometry}

The envelope in \eqref{eq:BGamma_def_journal} separates regimes according
to which transmission mode gives the largest possible total update
success. We first consider the case in which sensor~1 alone attains the
largest envelope value.

\begin{theorem}[Sensor-1 dominant frontier]
\label{thm:sensor1_dominant_frontier_journal}
Suppose that
\begin{align}
\Gamma_1s_1
\ge
\max\Big\{
&\Gamma_2s_2,
\nonumber\\
&
s_1\Gamma_1(1-\Gamma_2)
+s_2(1-\Gamma_1)\Gamma_2
+(a+b)\Gamma_1\Gamma_2
\Big\}.
\label{eq:sensor1_dominance_condition_journal}
\end{align}
Then the line segment
\begin{equation}
\mathcal L_1
=
\big\{
(q_1,q_2):q_1+q_2=\Gamma_1s_1,\;
q_1\ge0,\;q_2\ge0
\big\},
\label{eq:L1_def_journal}
\end{equation}
is an achievable Pareto frontier segment of $\mathcal Q_\Gamma$. It is
generated by the policy edge
\begin{equation}
v_1=v_2=0,\qquad
u_1=\Gamma_1\theta,\qquad
u_2=\Gamma_1(1-\theta),
\qquad 0\le\theta\le1.
\label{eq:sensor1_policy_edge_journal}
\end{equation}
If the inequality in \eqref{eq:sensor1_dominance_condition_journal} is
strict, then every Pareto-efficient point of $\mathcal Q_\Gamma$ lies on
$\mathcal L_1$.
\end{theorem}

\begin{IEEEproof}
By Theorem~\ref{thm:budget_success_envelope_journal}, every feasible
update-rate point satisfies $q_1+q_2\le \Gamma_1s_1$. The policy edge
\eqref{eq:sensor1_policy_edge_journal} gives
\[
q_1=\Gamma_1s_1\theta,\qquad
q_2=\Gamma_1s_1(1-\theta),
\]
and therefore generates the full segment $\mathcal L_1$.

Now take any feasible point $q=(q_1,q_2)$ with
$q_1+q_2<\Gamma_1s_1$. Let
\(\Delta=\Gamma_1s_1-q_1-q_2>0\).
Then $\widetilde q=(q_1+\Delta,q_2)$ lies on $\mathcal L_1$ and
dominates $q$. Hence no point strictly below the line
$q_1+q_2=\Gamma_1s_1$ can be Pareto efficient. No point on the line can
be dominated, because increasing one coordinate without decreasing the
other would violate the envelope. Therefore, $\mathcal L_1$ is a Pareto frontier segment. If \eqref{eq:sensor1_dominance_condition_journal} is
strict, no other envelope-attaining mode can generate an additional Pareto-efficient segment at the same total-success level.
\end{IEEEproof}

The sensor~2 dominant case is symmetric.

\subsection{Closed-Form Optimization on a Dominant Frontier}
\label{subsec:dominant_frontier_closed_form}

In the dominant-transmitter regimes of
Theorem~\ref{thm:sensor1_dominant_frontier_journal}, the independent
optimization reduces to a one-dimensional minimization over the dominant
Pareto segment. Let
\begin{equation}
S=
\begin{cases}
\Gamma_1s_1, & \text{sensor~1 dominant},\\
\Gamma_2s_2, & \text{sensor~2 dominant}.
\end{cases}
\label{eq:S_dominant_def_journal}
\end{equation}
Along the corresponding segment, write $q_1=S\theta$ and $q_2=S(1-\theta)$,
$0\le\theta\le1$. The restricted objective is
\begin{equation}
\phi(\theta)
=
w_1E_1(S\theta)+w_2E_2(S(1-\theta)).
\label{eq:dominant_phi_journal}
\end{equation}

\begin{theorem}[Closed-form dominant-frontier candidate]
\label{thm:dominant_closed_form_journal}
Let
$\eta\triangleq\sqrt{w_2K_2/(w_1K_1)}$.
If
\begin{equation}
S(\lambda_2+\eta\lambda_1)\ne0,
\label{eq:dominant_nonzero_condition_journal}
\end{equation}
then the unique interior stationary candidate of
\eqref{eq:dominant_phi_journal} is
\begin{equation}
\theta^\circ
=
\frac{
1-\lambda_2+\lambda_2S-\eta(1-\lambda_1)
}{
S(\lambda_2+\eta\lambda_1)
}.
\label{eq:theta_closed_form_journal}
\end{equation}
The global minimizer over the dominant frontier is obtained by evaluating
\begin{equation}
\theta\in
\{0,1\}
\cup
\{\theta^\circ:0<\theta^\circ<1\},
\label{eq:theta_candidate_set_journal}
\end{equation}
and selecting the point with the smallest value of $\phi(\theta)$.
\end{theorem}

\begin{IEEEproof}
Differentiating \eqref{eq:dominant_phi_journal} gives
$\phi'(\theta)=Sw_1E_1'(S\theta)-Sw_2E_2'(S(1-\theta))$.
Using \eqref{eq:Ei_first_derivative_journal},
an interior stationary point satisfies
\begin{equation}
\frac{w_1K_1}{D_1(S\theta)^2}
=
\frac{w_2K_2}{D_2(S(1-\theta))^2}.
\label{eq:dominant_balance_journal}
\end{equation}
Since $D_i(q)>0$, taking the positive square root gives
\begin{equation}
D_2(S(1-\theta))=\eta D_1(S\theta).
\label{eq:dominant_sqrt_balance_journal}
\end{equation}
Now
\[
D_1(S\theta)=1-\lambda_1+\lambda_1S\theta,
\]
and
\[
D_2(S(1-\theta))=1-\lambda_2+\lambda_2S(1-\theta).
\]
Substituting these affine expressions into
\eqref{eq:dominant_sqrt_balance_journal} and solving for $\theta$ gives
\eqref{eq:theta_closed_form_journal}. Since the minimum over the compact
interval $[0,1]$ is attained either at an endpoint or at an interior
stationary point, the candidate set in
\eqref{eq:theta_candidate_set_journal} is sufficient.
\end{IEEEproof}

\subsection{Cooperative Envelope Regime}
\label{subsec:cooperative_envelope_geometry}

The third envelope value in \eqref{eq:BGamma_def_journal} corresponds to
using both sampling budgets and exploiting simultaneous transmissions to
different sources. This cooperative envelope is
\begin{align}
B_\Gamma^{\mathrm{coop}}
&\triangleq
s_1\Gamma_1(1-\Gamma_2)
+s_2(1-\Gamma_1)\Gamma_2
+(a+b)\Gamma_1\Gamma_2.
\label{eq:BGamma_coop_journal}
\end{align}
It is achieved, for example, by $u_1=\Gamma_1$, $u_2=0$, $v_1=0$, $v_2=\Gamma_2$,
which gives
\begin{align}
q_1
&=
s_1\Gamma_1(1-\Gamma_2)+a\Gamma_1\Gamma_2,
\label{eq:coop_endpoint_12_q1}\\
q_2
&=
s_2(1-\Gamma_1)\Gamma_2+b\Gamma_1\Gamma_2.
\label{eq:coop_endpoint_12_q2}
\end{align}
The opposite different-source allocation $u_1=0$, $u_2=\Gamma_1$,
$v_1=\Gamma_2$, $v_2=0$ induces the mirrored pair, i.e.,
\eqref{eq:coop_endpoint_12_q1}--\eqref{eq:coop_endpoint_12_q2} with the
values of $q_1$ and $q_2$ interchanged.

When $B_\Gamma^{\mathrm{coop}}$ is the largest term in
\eqref{eq:BGamma_def_journal}, simultaneous different-source
transmission gives the largest total-success envelope. Unlike the
dominant-transmitter case, this does not by itself imply that the full
Pareto frontier is a single line segment. The structure of Pareto-efficient
independent randomized policies in this regime is treated in
Section~\ref{sec:pareto_reduction}.
\begin{remark}[Strong MPR and constrained operation]
Consider the full-budget case \(\Gamma_1=\Gamma_2=1\), and define \(M\triangleq\max\{s_1,s_2\}\) as the best single-sensor success probability. The cooperative different-source mode is useful only if its total simultaneous update success exceeds this best single-sensor term. In particular, if \(a+b<M\), the cooperative simultaneous-transmission term cannot be active. Hence, the dominant update-rate frontier is generated by a single sensor, sensor~1 if \(s_1\ge s_2\) and sensor~2 otherwise, and an optimal tradeoff can be achieved with one sensor silent.

Under sampling constraints, the corresponding condition becomes
budget dependent. In particular, the cooperative envelope is active
only if
\[
B_{\Gamma}^{\rm coop}
\geq
\max\{\Gamma_1s_1,\Gamma_2s_2\},
\]
or equivalently,
\[
s_2(1-\Gamma_1)+\Gamma_1(a+b-s_1)\geq 0,
\qquad
s_1(1-\Gamma_2)+\Gamma_2(a+b-s_2)\geq 0.
\]
If these inequalities fail, the total-success envelope is dominated
by a single-transmitter mode, and the Pareto-relevant operation is
generated by a policy in which one sensor remains silent while the
other allocates its budget between the two sources.
\end{remark}

\section{Pareto Reduction of Independent Randomization}
\label{sec:pareto_reduction}
The previous section characterized regimes in which the Pareto frontier
is generated by a dominant transmitter. We now consider the general
independent randomized problem under the sampling constraints
$t_1\leq \Gamma_1$ and $t_2\leq \Gamma_2$. The main result of this
section is that, for fixed transmission intensities, a Pareto-efficient
independent policy does not require both sensors to randomize between the
two sources. Equivalently, at least one sensor can be chosen to select a
single source whenever it transmits.

\subsection{Transmission Intensities and Conditional Source Selection}
\label{subsec:transmission_conditional_selection}

Recall that $t_1=u_1+u_2$ and $t_2=v_1+v_2$, where $0\leq t_1\leq\Gamma_1$ and
$0\leq t_2\leq\Gamma_2$. When $t_1,t_2>0$, define the conditional
source-selection probabilities
\begin{equation}
p\triangleq \frac{u_1}{t_1},
\qquad
q\triangleq \frac{v_1}{t_2}.
\label{eq:pq_conditional_pareto}
\end{equation}
Thus, $u_1=t_1p$, $u_2=t_1(1-p)$, $v_1=t_2q$, and $v_2=t_2(1-q)$,
with $p,q\in[0,1]$. The variables $t_1,t_2$ determine how often the
two sensors transmit, while $p,q$ determine which source is selected
conditional on transmission.

Define the sum and difference of the induced update probabilities as
\begin{equation}
Q\triangleq q_1+q_2,\qquad D\triangleq q_1-q_2.
\label{eq:QD_def_pareto}
\end{equation}
Then
\begin{equation}
q_1=\frac{Q+D}{2},\qquad q_2=\frac{Q-D}{2}.
\label{eq:q_from_QD}
\end{equation}
Hence, for a fixed value of $D$, increasing $Q$ increases both $q_1$
and $q_2$ simultaneously.

We further introduce $m\triangleq p+q-1$ and $r\triangleq p-q$, so that
$p=(1+m+r)/2$ and $q=(1+m-r)/2$.
The constraint $p,q\in[0,1]$ is equivalent to the diamond
\begin{equation}
|m|+|r|\leq 1.
\label{eq:diamond_constraint}
\end{equation}

For fixed $t_1,t_2$, the total update success can be written as
\begin{equation}
Q
=
C_0
+
\frac{t_1t_2}{2}
\left[
a+b+c
+
(a+b-c)(r^2-m^2)
\right],
\label{eq:Q_mr_pareto}
\end{equation}

where
\begin{equation}
C_0\triangleq s_1t_1(1-t_2)+s_2(1-t_1)t_2.
\label{eq:C0_def_pareto}
\end{equation}

Moreover, the imbalance $D=q_1-q_2$ is affine in $(m,r)$, $D=Am+Br$, with
\begin{equation}
\begin{aligned}
A&\triangleq
s_1t_1(1-t_2)+s_2(1-t_1)t_2+t_1t_2c,\\
B&\triangleq
s_1t_1(1-t_2)-s_2(1-t_1)t_2+t_1t_2(a-b).
\end{aligned}
\label{eq:AB_def_pareto}
\end{equation}

\subsection{Boundary Maximization at Fixed Imbalance}
\label{subsec:boundary_max_fixed_D}

The following lemma is the key geometric step. It shows that, for any
fixed feasible transmission intensities, the best way to improve the
update-rate pair without changing the split imbalance $D$ is to move to
the boundary of the conditional source-selection square.

\begin{lemma}[Boundary maximization at fixed imbalance]
\label{lem:boundary_fixed_D}
Assume $t_1,t_2>0$ and physically consistent channel
probabilities satisfying $a,b\geq0$ and
$\max\{a,b\}\leq c\leq a+b$.
For fixed $t_1,t_2$ and any feasible imbalance $D=q_1-q_2$,
the total update success $Q=q_1+q_2$ is maximized over
$p,q\in[0,1]$ at a point satisfying
\begin{equation}
p\in\{0,1\}
\quad\text{or}\quad
q\in\{0,1\}.
\label{eq:boundary_condition_pq}
\end{equation}
Equivalently, the maximizer lies on the boundary
$|m|+|r|=1$ of the diamond in \eqref{eq:diamond_constraint}.
\end{lemma}

\begin{IEEEproof}
For fixed $t_1,t_2$, the constraint of fixed imbalance is the line
\[
Am+Br=D.
\]
The feasible set for $(m,r)$ is the diamond $|m|+|r|\leq1$. Therefore,
the feasible set at fixed $D$ is either empty, a point, or a compact
line segment as shown in Fig.~\ref{fig:fixed_D_boundary}.
\begin{figure}[]
    \centering
    \includegraphics[width=0.55\linewidth]{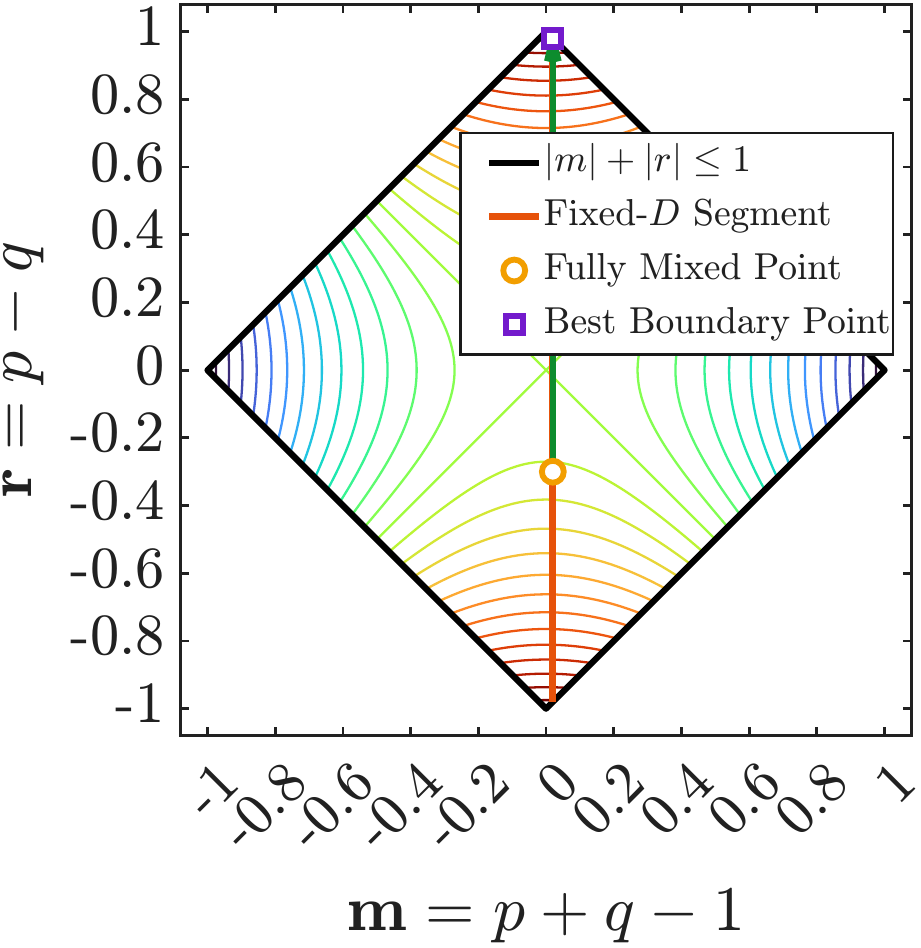}
    \caption{Fixed-$D$ segment and its intersection with the diamond constraint in \eqref{eq:diamond_constraint}.}
    \label{fig:fixed_D_boundary}
\end{figure}

First consider the case $a+b=c$. Then the coefficient
$a+b-c$ in \eqref{eq:Q_mr_pareto} vanishes,and hence \(Q=C_0+\frac{t_1t_2}{2}(a+b+c)\).
Therefore, all feasible points with the same imbalance $D$ have
the same total update success $Q$.

Because $D$ is feasible, the intersection of
$Am+Br=D$ with the diamond is nonempty. If $(A,B)\neq(0,0)$,
this intersection is a compact, possibly degenerate, line segment
whose endpoint or endpoints lie on the diamond boundary
$|m|+|r|=1$. If $A=B=0$, feasibility requires $D=0$, and any
point on the diamond boundary may be selected. Therefore, a boundary point exists with the same $D$ and the same $Q$.
Thus, the result holds when $a+b=c$, including the degenerate case
$A=|B|$.

In the remainder, assume $a+b>c$. Since
$c\geq\max\{a,b\}$, this strict inequality implies that
$a>0$ and $b>0$.

By \eqref{eq:Q_mr_pareto}, maximizing $Q$ over the fixed-$D$
feasible segment is then equivalent to maximizing
\[
r^2-m^2.
\]
We first show that $A>|B|$.
Since $c\geq\max\{a,b\}$, we have
$c+a-b\geq a>0$ and $c-a+b\geq b>0$. Therefore,
\begin{equation}
A+B
=
2s_1t_1(1-t_2)+t_1t_2(c+a-b)>0,
\label{eq:AplusB_general}
\end{equation}
and
\begin{equation}
A-B
=
2s_2(1-t_1)t_2+t_1t_2(c-a+b)>0.
\label{eq:AminusB_general}
\end{equation}
Thus, $A>|B|$.

If $B\neq0$, the line constraint gives
\[
r=\frac{D-Am}{B}.
\]
Substituting into $r^2-m^2$ yields
\[
\left(\frac{D-Am}{B}\right)^2-m^2,
\]
which is a strictly convex quadratic function of $m$ because
$A^2>B^2$. A strictly convex function attains its maximum over a
compact interval at an endpoint. If $B=0$, then the constraint fixes
$m=D/A$, and maximizing $r^2-m^2$ is equivalent to maximizing $r^2$,
which is also attained at an endpoint of the feasible segment.

Hence, in all cases, the maximum is attained at an endpoint of the
intersection between the line $Am+Br=D$ and the diamond
$|m|+|r|\leq1$. Such endpoints lie on the diamond boundary
$|m|+|r|=1$, which corresponds exactly to
$p=0$, $p=1$, $q=0$, or $q=1$.
\end{IEEEproof}

\subsection{No Fully Mixed Pareto-Optimal Source Selection}
\label{subsec:no_fully_mixed_pareto}

Lemma~\ref{lem:boundary_fixed_D} directly implies that fully mixed
conditional source selection is unnecessary for Pareto optimality.

\begin{theorem}[No fully mixed Pareto-optimal source selection]
\label{thm:no_fully_mixed_pareto}
Under the physically consistent channel conditions of
Lemma~\ref{lem:boundary_fixed_D},
consider any feasible independent policy satisfying
$0<t_1\leq\Gamma_1$, $0<t_2\leq\Gamma_2$, and
$p,q\in(0,1)$. Then there exists another feasible independent policy
with the same transmission intensities $t_1,t_2$, the same imbalance
$D=q_1-q_2$, and a weakly larger total update success $Q=q_1+q_2$.
Consequently, a Pareto-efficient independent policy can be chosen such
that
\begin{equation}
p\in\{0,1\}
\quad\text{or}\quad
q\in\{0,1\}.
\label{eq:no_fully_mixed_result}
\end{equation}
Thus, a fully mixed conditional source-selection policy is not required
for Pareto optimality under the sampling constraints.
\end{theorem}

\begin{IEEEproof}
Fix a feasible policy with $p,q\in(0,1)$ and transmission intensities
$t_1,t_2$. In the $(m,r)$ plane, this corresponds to an interior point
of the diamond. Fix the induced value of $D=q_1-q_2$. By
Lemma~\ref{lem:boundary_fixed_D}, there exists a boundary point
$(\widetilde m,\widetilde r)$ with the same $D$ and with
$\widetilde Q\geq Q$.

Since $D$ is unchanged, $\widetilde q_1=(\widetilde Q+D)/2$ and
$\widetilde q_2=(\widetilde Q-D)/2$. Therefore,
$\widetilde q_1\geq q_1$ and $\widetilde q_2\geq q_2$. The new policy has the same $t_1,t_2$, and hence still satisfies
$t_1\leq\Gamma_1$ and $t_2\leq\Gamma_2$. Since
$E_i(q_i)$ is strictly decreasing in $q_i$, the boundary policy weakly
improves the objective and strictly improves it whenever
$\widetilde Q>Q$. Hence no strictly Pareto-optimal policy requires
$p,q\in(0,1)$.
\end{IEEEproof}
\subsection{Reduction to Four Boundary Branches}
\label{subsec:four_boundary_branches}

Theorem~\ref{thm:no_fully_mixed_pareto} reduces the search over the
conditional source-selection square $(p,q)\in[0,1]^2$ to the four
boundary branches
\begin{equation}
p=1,\qquad p=0,\qquad q=1,\qquad q=0,
\label{eq:four_branches}
\end{equation}
on each of which one sensor is source-deterministic while the other
sensor may still randomize over its constrained action triangle. On a
fixed branch, with the transmission intensity of the
sensor fixed, the update-rate vector induced by
the other sensor is an affine image of its action probabilities, and
therefore lies in the convex hull of three vertex-induced update-rate
vectors $P^{(\cdot)}_0,P^{(\cdot)}_1,P^{(\cdot)}_2$. For each branch,
two of these three points lie on a common coordinate axis, so the
update-rate triangle is summarized by an axis point $A^{(\cdot)}$ and
an off-axis point $B^{(\cdot)}$. The detailed branch-by-branch
expressions for $P^{(\cdot)}_\ell$, $A^{(\cdot)}$, and $B^{(\cdot)}$,
together with the convex-hull representation, are collected in 
Appendix \ref{app:branch_expressions}; here we proceed directly to the branch-level Pareto
reduction.
\vspace{-13pt}

\subsection{Branch-Level Pareto Reduction}
\label{subsec:branch_level_reduction}

The previous branch descriptions show that, after fixing one boundary
branch and the transmission intensity of the source-deterministic
sensor, the remaining sensor induces a triangle in the update-rate plane.
The next result shows that only one edge of this triangle is relevant.

\begin{lemma}[Branch-level Pareto reduction]
\label{lem:branch_level_reduction}
Fix one of the four branches in \eqref{eq:four_branches} and fix the
transmission intensity of the sensor. Then the
other sensor can be chosen to randomize between at most two constrained
vertex actions. Equivalently, for fixed outer intensity, the
two-dimensional optimization over the other sensor's constrained action
triangle reduces to a one-dimensional optimization over the segment
joining the corresponding points $A$ and $B$.
\end{lemma}
\begin{figure}[]
    \centering
    \includegraphics[width=0.78\linewidth]{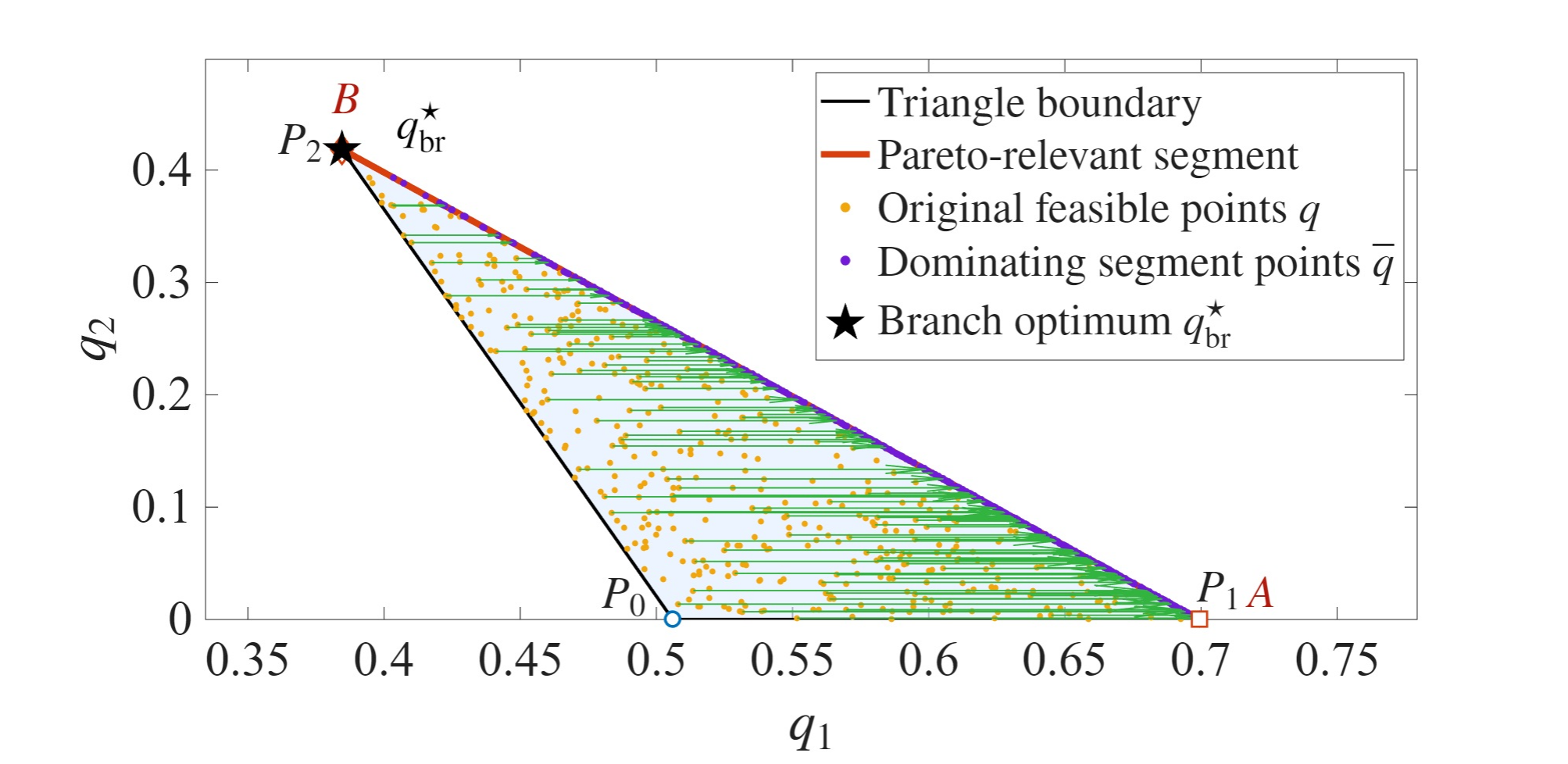}
\caption{Branch-level Pareto reduction for branch $p=1$. The feasible
triangle $\operatorname{conv}\{P_0,P_1,P_2\}$ reduces to the Pareto-relevant
segment $A$--$B$, on which the branch optimum $q^\star$ lies.}
    \label{fig:branch_pareto_reduction}
\end{figure}
\begin{IEEEproof}
For fixed outer intensity, the update-rate set generated by the other
sensor is the convex hull of three points
$\operatorname{conv}\{P_0,P_1,P_2\}$. In each of the four branches, two
of these points lie on the same coordinate axis. Any convex combination
of those two axis points is weakly dominated by the larger axis point
defined as $A$. Therefore, every point in the triangle is weakly
dominated by a point on the segment joining $A$ and the off-axis point
$B$. Since $F(q_1,q_2)$ is strictly decreasing in both coordinates,
a dominated point cannot be strictly optimal. Hence, for fixed outer
intensity, an optimal point can be chosen on the segment between $A$ and
$B$, which corresponds to randomizing between at most two constrained
vertex actions of the other sensor.
\end{IEEEproof}

\begin{lemma}[Closed-form candidate on a Pareto segment]
\label{lem:pareto_segment_candidate}
Let \(U=(U_1,U_2)\) and \(V=(V_1,V_2)\) be two Pareto-relevant
update-rate points ordered such that \(U_1\ge V_1\) and \(U_2\le V_2\).
Parameterize the segment between them as
\begin{equation}
q(\theta)=(1-\theta)U+\theta V,\qquad 0\le \theta\le 1 .
\label{eq:segment_param_general}
\end{equation}
Equivalently,
\begin{equation}
q_1(\theta)=U_1-L_1\theta,\qquad
q_2(\theta)=U_2+L_2\theta,
\label{eq:segment_q_affine}
\end{equation}
where \(L_1=U_1-V_1\ge0\) and \(L_2=V_2-U_2\ge0\). If
\(L_1=0\) or \(L_2=0\), the segment minimum is obtained by checking the
endpoints. Otherwise, define
\begin{equation}
\eta_{\mathrm e}
\triangleq
\sqrt{\frac{L_2w_2K_2}{L_1w_1K_1}},
\label{eq:eta_segment_candidate}
\end{equation}
and
\begin{equation}
d_1^U\triangleq 1-\lambda_1+\lambda_1U_1,\qquad
d_2^U\triangleq 1-\lambda_2+\lambda_2U_2 .
\label{eq:dU_segment_candidate}
\end{equation}
If \(\lambda_2L_2+\eta_{\mathrm e}\lambda_1L_1\neq0\), the only
interior stationary candidate is
\begin{equation}
\theta^\circ
=
\frac{\eta_{\mathrm e}d_1^U-d_2^U}
{\lambda_2L_2+\eta_{\mathrm e}\lambda_1L_1}.
\label{eq:theta_segment_candidate}
\end{equation}
Hence, the segment minimum is obtained by evaluating
\begin{equation}
\theta\in
\{0,1\}
\cup
\{\theta^\circ:0<\theta^\circ<1\}.
\label{eq:theta_segment_candidate_set}
\end{equation}
\end{lemma}

\begin{IEEEproof}
For \(L_1,L_2>0\), the restricted objective is
\[
\phi(\theta)=w_1E_1(q_1(\theta))+w_2E_2(q_2(\theta)).
\]
Using \(E_i'(q)=-K_i/D_i(q)^2\), an interior stationary point satisfies
\(L_1w_1K_1/D_1(q_1(\theta))^2=L_2w_2K_2/D_2(q_2(\theta))^2\).
Since \(D_i(q)>0\), taking the positive square root gives
\(D_2(q_2(\theta))=\eta_{\mathrm e}D_1(q_1(\theta))\).
Moreover,
\[
D_1(q_1(\theta))=d_1^U-\lambda_1L_1\theta,
\qquad
D_2(q_2(\theta))=d_2^U+\lambda_2L_2\theta .
\]
Solving the resulting affine equation gives
\eqref{eq:theta_segment_candidate}. Since the minimum on a compact
segment is attained either at an endpoint or at an interior stationary
point, the candidate set in \eqref{eq:theta_segment_candidate_set}
is sufficient.
\end{IEEEproof}
\vspace{-13pt}

\subsection{Branch Value Functions}
\label{subsec:branch_value_functions}

For each boundary branch and fixed outer transmission intensity,
Lemma~\ref{lem:branch_level_reduction} reduces the feasible update-rate
set to a Pareto-relevant segment with ordered endpoints \(U\) and \(V\).
The exact minimizer on this segment is obtained from
Lemma~\ref{lem:pareto_segment_candidate}.

For example, on the branch \(p=1\), let
\(U^{(p=1)}(x)\) and \(V^{(p=1)}(x)\) denote the ordered nondominated
endpoints generated at fixed \(x\in[0,\Gamma_1]\). The branch value is
\begin{equation}
\Phi_{p=1}(x)
=
\min_{0\leq\theta\leq1}
F\big((1-\theta)U^{(p=1)}(x)+\theta V^{(p=1)}(x)\big),
\label{eq:Phi_p1}
\end{equation}
where the inner minimum is evaluated using
Lemma~\ref{lem:pareto_segment_candidate}. The remaining branches are defined analogously, with
\(F_{p=r}^\star\triangleq\min_{0\leq x\leq\Gamma_1}\Phi_{p=r}(x)\)
and
\(F_{q=r}^\star\triangleq\min_{0\leq z\leq\Gamma_2}\Phi_{q=r}(z)\),
\(r\in\{0,1\}\). Hence, the independent optimum is
\begin{equation}
F_{\mathrm{ind}}^\star(\Gamma)
=
\min
\left\{
F_{p=1}^\star,\,
F_{p=0}^\star,\,
F_{q=1}^\star,\,
F_{q=0}^\star
\right\}.
\label{eq:Find_branch_min}
\end{equation}

Fig.~\ref{fig:branch_pareto_reduction} illustrates the branch-level reduction for
the branch $p=1$ at a fixed outer intensity. Although the remaining sensor
generates the full triangle $\operatorname{conv}\{P_0,P_1,P_2\}$, every
interior point $q$ is weakly dominated by a feasible replacement point
$\bar q$ on the segment $A$--$B$.

Consequently, the branch optimization does not need to search over the
two-dimensional triangle. It is sufficient to optimize over the
one-dimensional Pareto-relevant segment, where the optimal point
$q^\star$ is attained.
\vspace{-10pt}
\section{Coordinated MPR Time Sharing and Two-Mode Optimality}
\label{sec:coordinated_mpr}
\vspace{-5pt}

The independent policies studied in the previous sections impose a
factorization on the joint action probabilities. We now introduce a
coordinated time-sharing policy class in which a scheduler directly
selects the joint action of the two sensors.
\vspace{-18pt}

\subsection{Coordinated Joint-Action Policy}
\label{subsec:coordinated_policy}

Let \begin{equation}
\tau_{rs}\triangleq
\Pr\{a_1(t)=r,\;a_2(t)=s\},
\qquad r,s\in\{0,1,2\},
\label{eq:tau_def}
\end{equation}
where action $0$ denotes silence, action $1$ denotes transmission of
source $X_1$, and action $2$ denotes transmission of source $X_2$. The
coordinated policy $\tau=\{\tau_{rs}:r,s\in\{0,1,2\}\}$ is nonnegative and
satisfies a simplex constraint together with per-sensor sampling-budget
constraints, giving the  feasible set
\begin{equation}
\begin{aligned}
\mathcal T_\Gamma
\triangleq
\bigg\{
\tau\in\mathbb R_+^9:\;&
\sum_{r=0}^{2}\sum_{s=0}^{2}\tau_{rs}=1,\\
&
\sum_{r=1}^{2}\sum_{s=0}^{2}\tau_{rs}\leq\Gamma_1,\;
\sum_{r=0}^{2}\sum_{s=1}^{2}\tau_{rs}\leq\Gamma_2
\bigg\}.
\end{aligned}
\label{eq:Tgamma_coordinated}
\end{equation}

Each deterministic joint action $(r,s)$ induces an update-rate vector
$d^{rs}=(d_1^{rs},d_2^{rs})$. Using the channel notation already defined,
these vectors are

\begin{equation}
\begin{gathered}
d^{00}=(0,0),\quad
d^{10}=(s_1,0),\quad d^{20}=(0,s_1),\\
d^{01}=(s_2,0),\quad d^{02}=(0,s_2),\quad
d^{11}=(c,0),\\
d^{22}=(0,c),\quad
d^{12}=(a,b),\quad d^{21}=(b,a).
\end{gathered}
\label{eq:joint_action_vectors}
\end{equation}

Therefore, under a coordinated policy $\tau$, the effective update
probabilities are linear in $\tau$
\begin{equation}
q_i(\tau)
=
\sum_{r=0}^{2}\sum_{s=0}^{2}\tau_{rs}d_i^{rs},
\qquad i\in\{1,2\}.
\label{eq:q_tau_general}
\end{equation}
Equivalently,
\begin{equation}
q_1(\tau)
=
s_1\tau_{10}
+s_2\tau_{01}
+c\tau_{11}
+a\tau_{12}
+b\tau_{21},
\label{eq:q1_tau}
\end{equation}
\begin{equation}
q_2(\tau)
=
s_1\tau_{20}
+s_2\tau_{02}
+c\tau_{22}
+b\tau_{12}
+a\tau_{21}.
\label{eq:q2_tau}
\end{equation}

The coordinated sampling-constrained problem is
\begin{equation}
F_{\mathrm c}^{\star}(\Gamma)
=
\min_{\tau\in\mathcal T_\Gamma}
F(q_1(\tau),q_2(\tau)),
\label{eq:coordinated_problem_tau}
\end{equation}
where
\(F(q_1,q_2)=w_1E_1(q_1)+w_2E_2(q_2)\).
\vspace{-9pt}

\subsection{Coordinated Achievable Update-Rate Region}
\label{subsec:coordinated_region}

Define the coordinated achievable update-rate region as
\begin{equation}
\mathcal Q^c_\Gamma
\triangleq
\left\{
(q_1(\tau),q_2(\tau)):\tau\in\mathcal T_\Gamma
\right\}.
\label{eq:Qc_gamma}
\end{equation}
Then \eqref{eq:coordinated_problem_tau} can be written in the update-rate
plane as
\begin{equation}
F_{\mathrm c}^{\star}(\Gamma)
=
\min_{(q_1,q_2)\in\mathcal Q^c_\Gamma}
F(q_1,q_2).
\label{eq:coordinated_problem_q}
\end{equation}

\begin{lemma}[Coordinated update-rate polygon]
\label{lem:coordinated_polygon}
The set $\mathcal Q^c_\Gamma$ is a compact convex polygon in
$\mathbb R^2$.
\end{lemma}

\begin{IEEEproof}
The set $\mathcal T_\Gamma$ is a compact polytope because it is defined
by finitely many linear equalities and inequalities. The mapping
$\tau\mapsto(q_1(\tau),q_2(\tau))$ is linear by
\eqref{eq:q_tau_general}. Hence, $\mathcal Q^c_\Gamma$ is the linear
image of a compact polytope. Therefore, it is a compact convex polytope
in $\mathbb R^2$, i.e., a compact convex polygon.
\end{IEEEproof}

\begin{corollary}[Pareto-frontier localization]
\label{cor:coordinated_pareto_localization}
Assume $w_1,w_2>0$. Every global minimizer of
\eqref{eq:coordinated_problem_q} lies on the Pareto frontier of
$\mathcal Q^c_\Gamma$.
\end{corollary}

\begin{IEEEproof}
Suppose that $q^\star=(q_1^\star,q_2^\star)$ is a global minimizer and
is dominated by another point
$\widetilde q=(\widetilde q_1,\widetilde q_2)\in\mathcal Q^c_\Gamma$.
Then $\widetilde q_i\geq q_i^\star$ for $i\in\{1,2\}$, with at least
one strict inequality. Since each $E_i(q_i)$ is strictly decreasing,
$F(\widetilde q_1,\widetilde q_2)<F(q_1^\star,q_2^\star)$
which contradicts the optimality of $q^\star$. Hence no dominated point
can be globally optimal.
\end{IEEEproof}
\vspace{-11pt}
\subsection{Two-Mode Optimality}
\label{subsec:two_mode_optimality}

A Pareto-extreme point of $\mathcal Q^c_\Gamma$ is a vertex of
$\mathcal Q^c_\Gamma$ that lies on its Pareto frontier. For each
Pareto-extreme point, choose any feasible coordinated policy in
$\mathcal T_\Gamma$ that induces it; we refer to such a policy as a
Pareto-extreme coordinated mode.

\begin{theorem}[Two-mode optimality of coordinated MPR]
\label{thm:two_mode_optimality}
There exists a globally optimal coordinated policy whose induced
update-rate point lies either at a Pareto-extreme point of
$\mathcal Q^c_\Gamma$ or on a Pareto edge connecting two adjacent
Pareto-extreme points. Consequently, an optimal update-rate vector can be
achieved by time sharing between at most two Pareto-extreme coordinated
modes.
\end{theorem}

\begin{IEEEproof}
By Corollary~\ref{cor:coordinated_pareto_localization}, a global
minimizer lies on the Pareto frontier of $\mathcal Q^c_\Gamma$. By
Lemma~\ref{lem:coordinated_polygon}, $\mathcal Q^c_\Gamma$ is a compact
convex polygon. Every boundary point of a polygon is either a vertex or
lies on an edge connecting two adjacent vertices. Restricting this
statement to the nondominated boundary gives the claim for
Pareto-extreme points and Pareto edges.

If the optimal point lies on an edge between two Pareto-extreme points
$U$ and $V$, then it can be written as $(1-\theta)U+\theta V$ for some
$\theta\in[0,1]$. Let $\tau^U,\tau^V\in\mathcal T_\Gamma$ be feasible
coordinated policies that induce $U$ and $V$, respectively. Since
$\mathcal T_\Gamma$ is convex,
$\tau^\theta=(1-\theta)\tau^U+\theta\tau^V$
is feasible and induces $(1-\theta)U+\theta V$ by linearity of
\eqref{eq:q_tau_general}. Therefore, two Pareto-extreme coordinated
modes are sufficient.
\end{IEEEproof}

Theorem~\ref{thm:two_mode_optimality} reduces the coordinated
optimization to the Pareto vertices and Pareto edges of
\(\mathcal Q_\Gamma^c\). For any Pareto edge with ordered endpoints
\(U=(U_1,U_2)\) and \(V=(V_1,V_2)\), where \(U_1\ge V_1\) and
\(U_2\le V_2\), the edge minimizer is obtained directly from
Lemma~\ref{lem:pareto_segment_candidate}. Hence, the coordinated
optimization can be solved by evaluating all Pareto vertices and the
closed-form candidate \(\theta^\circ\) on each Pareto edge.

\vspace{-16pt}

\subsection{Pareto-Edge Search}
\label{subsec:finite_pareto_edge_search}

The previous results lead to a finite search over the coordinated Pareto
frontier.

\begin{algorithm}[t]
\caption{Pareto-Edge Search for Coordinated MPR Time Sharing}
\label{alg:coordinated_pareto_edge_search}
\begin{algorithmic}[1]
\STATE \textbf{Input:} Source parameters $\alpha_i,\beta_i,w_i$, channel parameters $s_1,s_2,a,b,c$, and budgets $\Gamma_1,\Gamma_2$.
\STATE Construct the coordinated feasible set $\mathcal T_\Gamma$ in \eqref{eq:Tgamma_coordinated}.
\STATE Compute the coordinated update-rate polygon $\mathcal Q^c_\Gamma$ in \eqref{eq:Qc_gamma}.
\STATE Extract the Pareto vertices $\mathcal V_{\mathrm P}^c$ and Pareto edges $\mathcal E_{\mathrm P}^c$ of $\mathcal Q^c_\Gamma$.
\STATE Initialize $\mathcal C\leftarrow\mathcal V_{\mathrm P}^c$.
\FOR{each Pareto edge $(U,V)\in\mathcal E_{\mathrm P}^c$}
    \STATE Order the endpoints such that $U_1\geq V_1$ and $U_2\leq V_2$.
    \STATE Compute $L_1=U_1-V_1$ and $L_2=V_2-U_2$.
    \IF{$L_1>0$ and $L_2>0$}

\STATE Compute \(\eta_{\mathrm e}\) from \eqref{eq:eta_segment_candidate}.
\STATE Compute \(\theta^\circ\) from \eqref{eq:theta_segment_candidate}.

        \IF{$0<\theta^\circ<1$}
            \STATE Add $q^\circ=(1-\theta^\circ)U+\theta^\circ V$ to $\mathcal C$.
        \ENDIF
    \ENDIF
\ENDFOR
\STATE Select $q^\star\in\arg\min_{q\in\mathcal C}F(q_1,q_2)$.
\STATE \textbf{Output:} The optimal coordinated update-rate vector $q^\star$ and a corresponding coordinated time-sharing policy.
\end{algorithmic}
\end{algorithm}

\begin{theorem}[Global optimality of Pareto-edge search]
\label{thm:pareto_edge_search_global}
Algorithm~\ref{alg:coordinated_pareto_edge_search} returns a global
minimizer of the problem \eqref{eq:coordinated_problem_q}.
\end{theorem}

\begin{IEEEproof}
By Corollary~\ref{cor:coordinated_pareto_localization}, a global
minimizer lies on the Pareto frontier of $\mathcal Q^c_\Gamma$. Since
$\mathcal Q^c_\Gamma$ is a polygon, this frontier consists of finitely
many vertices and edges. Algorithm~\ref{alg:coordinated_pareto_edge_search}
includes all Pareto vertices in the candidate set. On each Pareto edge,
Lemma~\ref{lem:pareto_segment_candidate} shows that the minimum is
attained either at an endpoint or at the feasible interior stationary
candidate. The endpoints are already included as Pareto vertices, and the
algorithm adds every feasible interior candidate. Therefore, minimizing
$F(q_1,q_2)$ over the constructed candidate set gives a global minimizer
over $\mathcal Q^c_\Gamma$, and hence over $\mathcal T_\Gamma$.
\end{IEEEproof}

The coordinated formulation therefore reduces the MPR scheduling
benchmark to a finite geometric search over a convex update-rate polygon.
\vspace{-12pt}

\section{Independent--Coordinated Gap}
\label{sec:coordination_gap}
\vspace{-3pt}

We compare independent randomized policies with coordinated time sharing under the same sampling budgets. Because independent randomization restricts joint-action probabilities to a product form, the independent class is a subset of the coordinated class.

\begin{lemma}[Embedding of independent policies into coordinated policies]
\label{lem:independent_embedded_coordinated}
For every feasible independent policy in $\mathcal A_\Gamma$, there
exists a feasible coordinated policy in $\mathcal T_\Gamma$ that induces
the same update-rate pair $(q_1,q_2)$. Consequently,
\begin{equation}
\mathcal Q_\Gamma \subseteq \mathcal Q^c_\Gamma .
\label{eq:Qind_subset_Qc}
\end{equation}
\end{lemma}

\begin{IEEEproof}
Take any feasible independent policy
$(u_1,u_2,v_1,v_2)\in\mathcal A_\Gamma$ and define
\[
a_{1,0}=1-u_1-u_2,\qquad a_{1,1}=u_1,\qquad a_{1,2}=u_2,
\]
and
\[
a_{2,0}=1-v_1-v_2,\qquad a_{2,1}=v_1,\qquad a_{2,2}=v_2.
\]
Construct a coordinated policy by
\begin{equation}
\tau_{rs}=a_{1,r}a_{2,s},
\qquad r,s\in\{0,1,2\}.
\label{eq:tau_product_embedding}
\end{equation}
Then $\tau_{rs}\geq0$ and
$\sum_{r=0}^{2}\sum_{s=0}^{2}\tau_{rs}=\big(\sum_{r=0}^{2}a_{1,r}\big)\big(\sum_{s=0}^{2}a_{2,s}\big)=1$.
Moreover,
$\sum_{r=1}^{2}\sum_{s=0}^{2}\tau_{rs}=a_{1,1}+a_{1,2}=u_1+u_2\leq \Gamma_1$
and
$\sum_{r=0}^{2}\sum_{s=1}^{2}\tau_{rs}=a_{2,1}+a_{2,2}=v_1+v_2\leq \Gamma_2$.
Hence $\tau\in\mathcal T_\Gamma$.
Finally, substituting
\eqref{eq:tau_product_embedding} into
\eqref{eq:q1_tau}--\eqref{eq:q2_tau} gives exactly the independent
update probabilities in
\eqref{eq:q1_bilinear_ind}--\eqref{eq:q2_bilinear_ind}. Therefore, every
update-rate pair achievable by an independent policy is achievable
by a coordinated policy.
\end{IEEEproof}

\begin{theorem}[Coordinated lower bound]
\label{thm:coordinated_lower_bound}
The coordinated optimum is a lower bound on the independent randomized
optimum
\begin{equation}
F_{\mathrm c}^{\star}(\Gamma)
\leq
F_{\mathrm{ind}}^{\star}(\Gamma).
\label{eq:coordinated_lower_bound}
\end{equation}
\end{theorem}

\begin{IEEEproof}
By Lemma~\ref{lem:independent_embedded_coordinated},
$\mathcal Q_\Gamma\subseteq\mathcal Q^c_\Gamma$. Since both optimization
problems minimize the same objective $F(q_1,q_2)$ over their respective
update-rate regions, minimizing over the larger set
$\mathcal Q^c_\Gamma$ cannot give a larger value than minimizing over
the smaller set $\mathcal Q_\Gamma$. Hence
\eqref{eq:coordinated_lower_bound} follows.
\end{IEEEproof}

The independent--coordinated gap is defined as
\begin{equation}
\Delta_\Gamma
\triangleq
F_{\mathrm{ind}}^{\star}(\Gamma)
-
F_{\mathrm c}^{\star}(\Gamma).
\label{eq:coordination_gap_def}
\end{equation}
By Theorem~\ref{thm:coordinated_lower_bound}, $\Delta_\Gamma\geq0$.
This gap quantifies the performance loss caused by restricting the two
sensors to act independently.

\begin{corollary}[Zero-gap condition]
\label{cor:zero_gap_condition}
The gap satisfies $\Delta_\Gamma=0$ if and only if at least one
coordinated optimal update-rate point is also achievable by a feasible
independent randomized policy, i.e.,
\begin{equation}
\arg\min_{q\in\mathcal Q^c_\Gamma}F(q_1,q_2)
\cap
\mathcal Q_\Gamma
\neq\emptyset.
\label{eq:zero_gap_condition}
\end{equation}
\end{corollary}

\begin{IEEEproof}
If a coordinated optimal point belongs to $\mathcal Q_\Gamma$, then the
independent problem can achieve the coordinated optimal value, so
$F_{\mathrm{ind}}^{\star}(\Gamma)=F_{\mathrm c}^{\star}(\Gamma)$.
Conversely, if $\Delta_\Gamma=0$, then the minimum value over
$\mathcal Q_\Gamma$ equals the minimum value over $\mathcal Q^c_\Gamma$.
Since $\mathcal Q_\Gamma\subseteq\mathcal Q^c_\Gamma$, an independent
optimal point is also coordinated-optimal, proving
\eqref{eq:zero_gap_condition}.
\end{IEEEproof}

\begin{remark}
The condition in Corollary~\ref{cor:zero_gap_condition} is expressed in
the update-rate plane. It does not require the coordinated optimizer
itself to have a product-form representation in $\tau$; it only requires
that the same optimal update-rate pair can be generated by some feasible
independent policy. If no independent policy can generate any coordinated
optimal update-rate point, then the gap is strictly positive.
\end{remark}

Thus, the coordinated formulation provides a benchmark for the best
sampling-constrained MPR performance, while
$\Delta_\Gamma$ measures the cost of decentralized independent
randomization.
\vspace{-13pt}

\section{Numerical Results}
\label{sec:numerical_results}
\vspace{-4pt}

We evaluate the proposed independent and coordinated MPR sampling schemes using the parameter set in Table~\ref{tab:numerical_parameters}.

\begin{table}[t]
\centering
\caption{Default numerical parameters.}
\label{tab:numerical_parameters}
\begin{tabular}{c|c}
\hline
Parameter & Value \\
\hline
$(\alpha_1,\beta_1)$ & $(0.52,0.18)$ \\
$(\alpha_2,\beta_2)$ & $(0.45,0.22)$ \\
$(\lambda_1,\lambda_2)$ & $(0.30,0.33)$ \\
$(s_1,s_2)$ & $(0.92,0.82)$ \\
$(a,b)$ & $(0.58,0.50)$ \\
$c$  & $0.79$ \\
$(\Gamma_1,\Gamma_2)$ & $(0.65,0.65)$ \\
$(w_1,w_2)$ & $(0.5,0.5)$ \\
\hline
\end{tabular}
\end{table}
\subsection{Geometry of the Update-Rate Regions}
\label{subsec:numerical_update_rate_geometry}

Fig.~\ref{fig:update_rate_geometry_numerical} illustrates the achievable
update-rate regions under the parameters in Table~\ref{tab:numerical_parameters}. Since the objective is decreasing in both
\(q_1\) and \(q_2\), the optimum lies on the Pareto frontier. The
coordinated region contains the independent region and therefore reaches
a lower objective contour, showing the benefit of enlarging the feasible
update-rate tradeoff through coordination.

\begin{figure}[t]
    \centering
    \includegraphics[width=0.65\columnwidth]{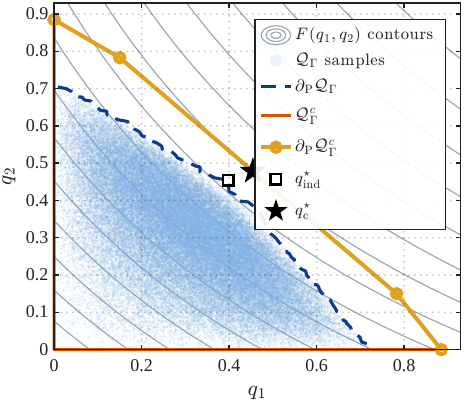}
    \caption{Achievable update-rate regions under sampling constraints.}
    \label{fig:update_rate_geometry_numerical}
\end{figure}
\vspace{-10px}
\subsection{Regimes and Dominant Frontiers}
We next illustrate the budget-aware total-success envelope
$B_\Gamma=\max\{\Gamma_1s_1,\Gamma_2s_2,B_{\rm coop}\}$
which separates the independent update-rate geometry into different
structural regimes. This
experiment uses the same parameters in
Table~\ref{tab:numerical_parameters},
but reduces the simultaneous-transmission probabilities to
$a=0.12$ and $b=0.10$. Fig.~\ref{fig:envelope_regimes} confirms the three cases predicted by the
analysis. In the sensor-1 dominant region, the largest envelope value is
$\Gamma_1s_1$, and the Pareto frontier is generated by allocating the
available updates through sensor~1. In the sensor-2 dominant region, the
largest value is $\Gamma_2s_2$, and the symmetric dominant-frontier result
applies. In the cooperative region, the largest value is $B_{\rm coop}$,
which means that using both budgets for simultaneous different-source
transmissions provides the largest total-success potential.

\begin{figure}[t]
    \centering
    \includegraphics[width=0.75\columnwidth]{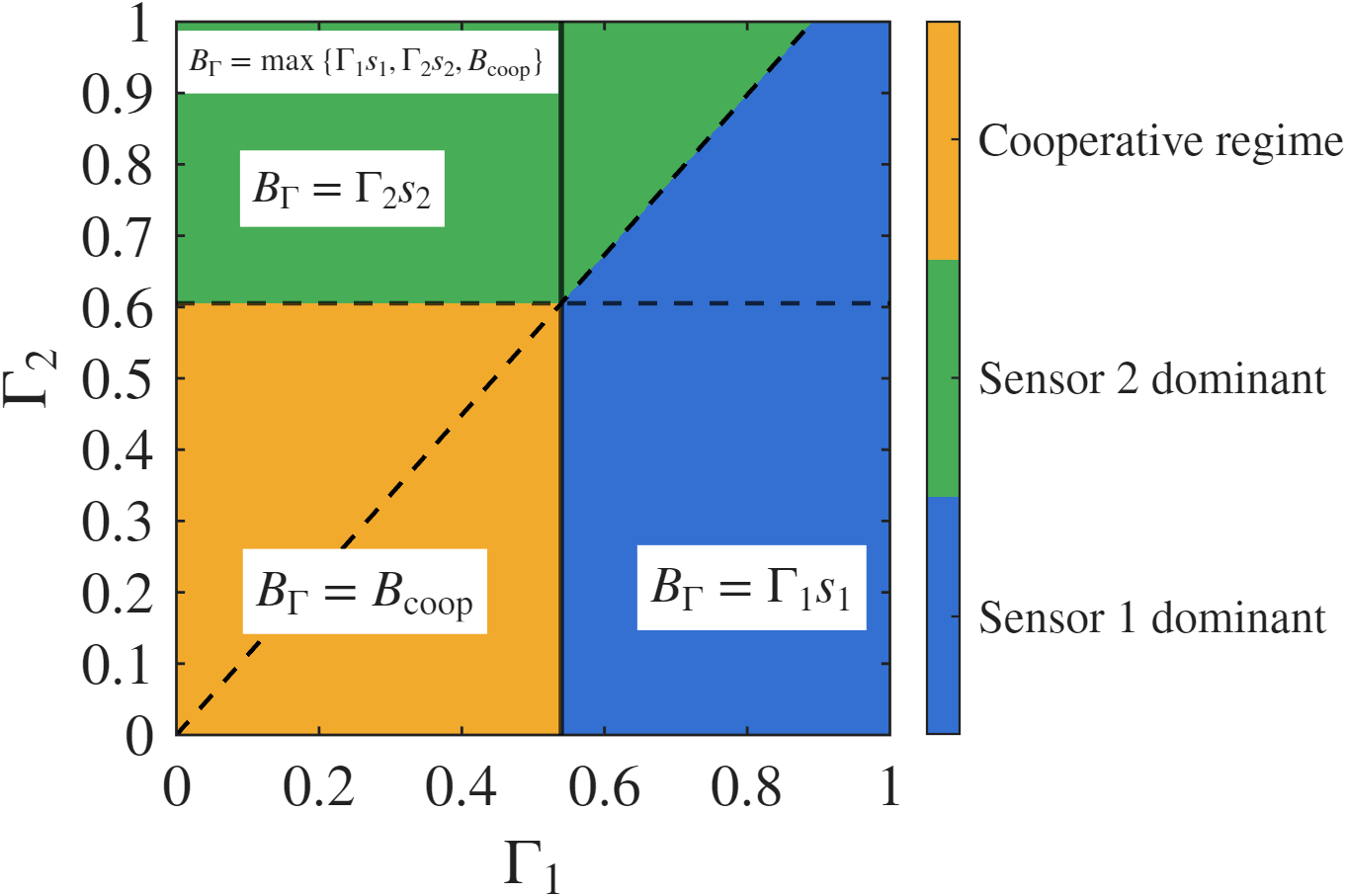}
    \caption{Regimes induced by the budget-aware 
    bound. }
    \label{fig:envelope_regimes}
\end{figure}

\vspace{-10pt}
\subsection{Policy-Class Comparison Across MPR Regimes and Sampling Budgets}
\label{subsec:mpr_regime_comparison}

We compare the considered policy classes under three physical-layer regimes:
collision, capture, and MPR. The regimes differ only in the simultaneous-transmission
success probabilities. In the collision channel, no packet is decoded under
simultaneous transmission, i.e.,
\(a=b=0, c=0.\)
In the capture channel, only one of the two simultaneous packets can be decoded,
with
\(a=0.58,\ b=0,\ \text{and } c=0.58.\)
In the MPR channel, both simultaneous packets can be decoded with nonzero
probability, with
\(a=0.58,\ b=0.50,\ \text{and } c=0.79.\)

Both the collision and capture channels satisfy
$a+b=c$. Hence, the boundary-case argument in
Lemma~\ref{lem:boundary_fixed_D} applies, including the possible
degeneracy $A=|B|$. Consequently,
Theorem~\ref{thm:no_fully_mixed_pareto} and the four-branch reduction
in \eqref{eq:four_branches} remain valid for both regimes. The optimized
independent collision and capture policies reported below are therefore
obtained by minimizing over the same four boundary branches used for
the general MPR channel. We compare four policy classes. The uniform independent baseline assigns each
sensor's sampling budget equally to the two sources, i.e.,
\(u_1=u_2=\Gamma_1/2\) and \(v_1=v_2=\Gamma_2/2\), without coordination
between the sensors. The optimized TDMA baseline is the best single-active
coordinated policy, where time sharing is restricted to the five joint actions
\((0,0)\), \((1,0)\), \((2,0)\), \((0,1)\), and \((0,2)\), so simultaneous
transmissions are not allowed. Equivalently, the sensors are orthogonalized in
time and the long-term fractions \(\rho_{11},\rho_{12},\rho_{21},\rho_{22}\)
are optimized, where \(\rho_{ki}\) denotes the fraction of slots in which
sensor \(k\) transmits source \(X_i\). These fractions satisfy
\(\rho_{11}+\rho_{12}\leq\Gamma_1,\qquad \rho_{21}+\rho_{22}\leq\Gamma_2\)
together with the TDMA orthogonality constraint
\(\rho_{11}+\rho_{12}+\rho_{21}+\rho_{22}\leq 1\).
The optimized independent
MPR policy solves the independent randomized problem over
\((u_1,u_2,v_1,v_2)\) using the MPR update-rate expressions. Finally, the
coordinated MPR benchmark allows time sharing over all nine joint actions,
including simultaneous transmissions.

\begin{figure*}[t]
    \centering
    \includegraphics[width=0.55\textwidth]{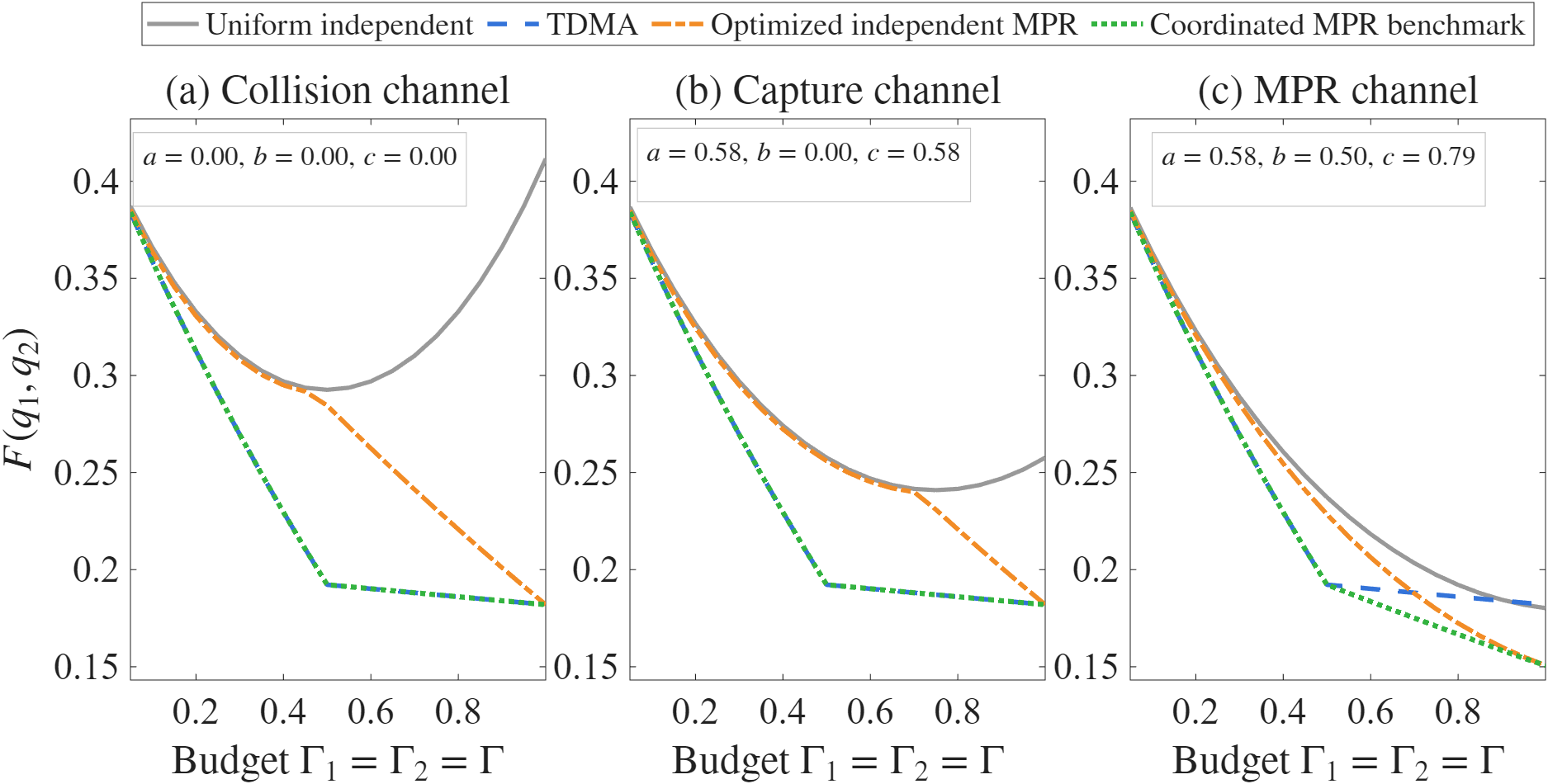}
    \caption{Policy comparison across collision, capture, and MPR regimes versus
    common sampling budget.}
    \label{fig:mpr_regime_budget_comparison}
\end{figure*}

Fig.~\ref{fig:mpr_regime_budget_comparison} studies the effect of the
common sampling budget \(\Gamma_1=\Gamma_2=\Gamma\). In the collision channel,
the coordinated MPR benchmark coincides with TDMA because overlapping
transmissions do not yield useful receptions, while the uniform independent
policy degrades at large budgets due to the increased number of collisions.
The capture channel shows a similar behavior: although one packet may be
decoded during simultaneous transmission, the resulting update-rate tradeoff
does not improve the task-oriented objective over optimized single-active
scheduling. In the MPR channel, both the optimized independent MPR policy and
the coordinated MPR benchmark outperform TDMA at larger budgets. At small
budgets, overlap is rarely needed and the policies remain close. As the budget
increases, single-active access cannot fully exploit the available sampling
opportunities, whereas MPR can use simultaneous transmissions. Thus, the
benefit of MPR appears only when the physical layer supports useful concurrent
reception for both sources.

\vspace{-10pt}
\section{Conclusion}
\vspace{-6pt}

\label{sec:conclusion}
We studied the real-time reconstruction and remote actuation of two binary Markov  sources over a shared MPR channel. We derived closed-form expressions for the RTE and CAE under synchronize-or-hold estimation and showed that, for the considered binary-source model, CAE minimization reduces to a weighted RTE minimization problem. We then characterized the sampling-constrained independent randomized problem through the geometry of the achievable update-rate region. The analysis shows that, although the MPR-induced update-rate map is bilinear and the optimization is nonconvex, Pareto-efficient independent policies can be searched through a finite set of one-dimensional boundary branches. We also introduced a coordinated time-sharing benchmark and proved that an optimal coordinated policy requires time sharing between at most two Pareto-extreme modes. Numerical results show that MPR does not automatically improve the task-oriented objective. In collision and capture channel, simultaneous transmissions provide no gain over single-active scheduling. A clear improvement appears only when concurrent reception is sufficiently reliable for both sources, especially at larger sampling budgets where orthogonal access cannot fully exploit the available transmission opportunities.
\vspace{-11pt}

\appendices

\section{Branch-Level Update-Rate Expressions}
\label{app:branch_expressions}
\vspace{-5pt}

This appendix collects the detailed update-rate expressions for the
four boundary branches in \eqref{eq:four_branches}, which underlie the
branch-level Pareto reduction of Section~\ref{sec:pareto_reduction}.

Theorem~\ref{thm:no_fully_mixed_pareto} reduces the search over the
conditional source-selection square $(p,q)\in[0,1]^2$ to the four
boundary branches. On each branch, one sensor is source-deterministic, while the other
sensor may still randomize over its constrained action.

For each branch, let $P^{(\cdot)}_\ell(\cdot)\in\mathbb R^2$,
$\ell\in\{0,1,2\}$, denote the update-rate vector $(q_1,q_2)$ induced
by the $\ell$th constrained vertex action of the non-deterministic
sensor. In the branches $p=1$ and $p=0$, the fixed intensity is denoted
by $x\in[0,\Gamma_1]$, and $\ell=0,1,2$ correspond to the sensor~2 vertices are \((v_1,v_2)=(0,0),\quad (\Gamma_2,0),\quad (0,\Gamma_2)\).
In the branches $q=1$ and $q=0$, the fixed intensity is denoted by
$z\in[0,\Gamma_2]$, and $\ell=0,1,2$ correspond to the sensor. The vertices are \((u_1,u_2)=(0,0),\quad (\Gamma_1,0),\quad (0,\Gamma_1)\).

\begin{lemma}[Convex-hull representation on a fixed branch]
\label{lem:branch_convex_hull}
Fix one of the four boundary branches in \eqref{eq:four_branches} and
fix the transmission intensity of the sensor with deterministic source selection. The
update-rate vector induced by any feasible randomized action of the
other sensor belongs to the convex hull of the three
vertex-induced update-rate vectors.
\end{lemma}

\begin{IEEEproof}
On a fixed branch, the policy variables of the source-deterministic
sensor are fixed. Hence, by
\eqref{eq:q1_bilinear_ind}--\eqref{eq:q2_bilinear_ind}, the mapping from
the other sensor's two action probabilities to $(q_1,q_2)$ is affine.

Consider the case where sensor~2 is the non-deterministic sensor. Any
feasible $(v_1,v_2)$ satisfying $v_1,v_2\ge0$ and
$v_1+v_2\le\Gamma_2$ can be written as
\begin{equation}
(v_1,v_2)
=
\mu_0(0,0)+\mu_1(\Gamma_2,0)+\mu_2(0,\Gamma_2),
\label{eq:v_convex_branch}
\end{equation}
where
\begin{equation}
\mu_0=1-\frac{v_1+v_2}{\Gamma_2},\qquad
\mu_1=\frac{v_1}{\Gamma_2},\qquad
\mu_2=\frac{v_2}{\Gamma_2}.
\label{eq:v_mu_branch}
\end{equation}
The coefficients are nonnegative and satisfy
$\mu_0+\mu_1+\mu_2=1$. Since affine mappings preserve convex
combinations, the induced update-rate vector satisfies
\begin{equation}
(q_1,q_2)
=
\mu_0P^{(\cdot)}_0+\mu_1P^{(\cdot)}_1+\mu_2P^{(\cdot)}_2.
\label{eq:generic_convex_P}
\end{equation}
Thus, it belongs to the convex hull of the three vertex-induced
update-rate vectors. The same argument applies when sensor~1 is the
non-deterministic sensor, using the constrained vertices
$(0,0)$, $(\Gamma_1,0)$, and $(0,\Gamma_1)$.
\end{IEEEproof}

\subsubsection{Branch $p=1$}

In this branch,
\begin{equation}
u_1=x,\qquad u_2=0,\qquad 0\le x\le\Gamma_1.
\label{eq:branch_p1_policy}
\end{equation}
According to the vertex indexing defined above, the three update-rate
vectors are
\begin{align}
P^{(p=1)}_0(x)
&=
(s_1x,0),
\label{eq:Pp1_0}\\
P^{(p=1)}_1(x)
&=
\big(s_1x(1-\Gamma_2)+s_2(1-x)\Gamma_2+cx\Gamma_2,\;0\big),
\label{eq:Pp1_1}\\
P^{(p=1)}_2(x)
&=
\big(s_1x(1-\Gamma_2)+ax\Gamma_2,\;
s_2(1-x)\Gamma_2+bx\Gamma_2\big).
\label{eq:Pp1_2}
\end{align}
By Lemma~\ref{lem:branch_convex_hull}, for fixed $x$,
\begin{equation}
(q_1,q_2)\in
\operatorname{conv}
\left\{
P^{(p=1)}_0(x),
P^{(p=1)}_1(x),
P^{(p=1)}_2(x)
\right\}.
\label{eq:p1_convex_hull}
\end{equation}
Since $P^{(p=1)}_0(x)$ and $P^{(p=1)}_1(x)$ lie on the $q_1$-axis, only
their coordinate-wise dominant point is Pareto relevant. Define

\begin{equation}
\begin{aligned}
A^{(p=1)}(x)
&=
\Big(
\max\big\{
\big[P^{(p=1)}_{0}(x)\big]_1,
\big[P^{(p=1)}_{1}(x)\big]_1
\big\},\,0
\Big),\\
B^{(p=1)}(x)
&=
P^{(p=1)}_2(x).
\end{aligned}
\label{eq:Ap1_def}
\end{equation}

\subsubsection{Branch $p=0$}

Here sensor~1 transmits source $X_2$ whenever it is active
\begin{equation}
u_1=0,\qquad u_2=x,\qquad 0\leq x\leq\Gamma_1.
\label{eq:branch_p0_policy}
\end{equation}
For fixed $x$, the three constrained vertices of sensor~2 induce
\begin{align}
P^{(p=0)}_0(x)
&=
(0,s_1x),
\label{eq:Pp0_0}\\
P^{(p=0)}_1(x)
&=
\big(s_2(1-x)\Gamma_2+bx\Gamma_2,\;
s_1x(1-\Gamma_2)+ax\Gamma_2\big),
\label{eq:Pp0_1}\\
P^{(p=0)}_2(x)
&=
\big(0,\;
s_1x(1-\Gamma_2)+s_2(1-x)\Gamma_2+cx\Gamma_2
\big).
\label{eq:Pp0_2}
\end{align}
The first and third points lie on the $q_2$-axis. Define

\begin{equation}
\begin{aligned}
A^{(p=0)}(x)
&=
\Big(
0,\,
\max\big\{
\big[P^{(p=0)}_{0}(x)\big]_2,
\big[P^{(p=0)}_{2}(x)\big]_2
\big\}
\Big),\\
B^{(p=0)}(x)
&=
P^{(p=0)}_1(x).
\end{aligned}
\label{eq:Ap0_def}
\end{equation}

\subsubsection{Branch $q=1$}

Here sensor~2 transmits source $X_1$ whenever it is active
\begin{equation}
v_1=z,\qquad v_2=0,\qquad 0\leq z\leq\Gamma_2.
\label{eq:branch_q1_policy}
\end{equation}
For fixed $z$, the three constrained vertices of sensor~1 induce
\begin{align}
P^{(q=1)}_0(z)
&=
(s_2z,0),
\label{eq:Pq1_0}\\
P^{(q=1)}_1(z)
&=
\big(s_2z(1-\Gamma_1)+s_1(1-z)\Gamma_1+c\Gamma_1z,\;0\big),
\label{eq:Pq1_1}\\
P^{(q=1)}_2(z)
&=
\big(s_2z(1-\Gamma_1)+b\Gamma_1z,\;
s_1(1-z)\Gamma_1+a\Gamma_1z\big).
\label{eq:Pq1_2}
\end{align}
The first two points lie on the $q_1$-axis. Define
\begin{equation} \begin{aligned} A^{(q=1)}(z) &= \Big( \max\big\{ \big[P^{(q=1)}_{0}(z)\big]_1, \big[P^{(q=1)}_{1}(z)\big]_1 \big\},\,0 \Big),\\ B^{(q=1)}(z) &= P^{(q=1)}_2(z). \end{aligned} \label{eq:Aq1_def} \end{equation}
\subsubsection{Branch \(q=0\)} This branch is symmetric to the branch \(q=1\) under the interchange of the source labels \(1\) and \(2\). In particular, sensor~2 transmits source \(X_2\) whenever it is active, i.e., \(v_1=0\) and \(v_2=z\), with \(0\le z\le \Gamma_2\). The corresponding vertex-induced points are obtained from \eqref{eq:Pq1_0}--\eqref{eq:Aq1_def} by swapping the two update-rate coordinates and interchanging the roles of source \(X_1\) and source \(X_2\).

\vspace{-12pt}
\bibliographystyle{IEEEtran}
\bibliography{Refs_cleaned}

\end{document}